\documentclass[lettersize,journal]{IEEEtran}
\usepackage{amsmath,amsfonts}
\usepackage{algorithmic}
\usepackage{algorithm}
\usepackage{array}
\usepackage[caption=false,font=normalsize,labelfont=sf,textfont=sf]{subfig}
\usepackage{textcomp}
\usepackage{stfloats}
\usepackage{url}
\usepackage{verbatim}
\usepackage{graphicx}
\usepackage{cite}
\usepackage{amssymb}
\usepackage{verbatim}
\usepackage{amsthm}
\newtheorem{theorem}{Theorem}[section]
\newtheorem{definition}[theorem]{Definition}

\newtheorem{lemma}[theorem]{Lemma}
\newtheorem{remark}[theorem]{Remark}

\newtheorem{proposition}[theorem]{Proposition}
\makeatletter

\newcommand{\Rmnum}[1]{\expandafter\@slowromancap\romannumeral #1@}
\makeatother

\begin{document}

\title{Reconstruction with prior support information and non-Gaussian constraints}

\author{
\IEEEauthorblockN{
Xiaotong~Liu$^1$,
~Yiyu~Liang$^{2*}$}

\IEEEauthorblockA{
$^{1,2}$School of Mathematics and Statistics,
Beijing Jiaotong University, Beijing 100044,
PEOPLE’S REPUBLIC OF CHINA}
\thanks{Yiyu Liang is supported by the National Natural Science Foundation of China (Grant Nos. 12271021).}}




\maketitle
\begin{abstract}

In this study, we introduce a novel model, termed the Weighted Basis Pursuit Dequantization ($\omega
$-BPDQ$_p$), which incorporates prior support information by assigning weights on the $\ell_1$ norm in the $\ell_1$ minimization process and replaces the $\ell_2$ norm with the $\ell_p$ norm in the constraint. This adjustment addresses cases where noise deviates from a Gaussian distribution, such as quantized errors, which are common in practice. We demonstrate that Restricted Isometry Property (RIP$_{p,q}$) and Weighted Robust Null Space Property ($\omega$-RNSP$_{p,q}$) ensure stable and robust reconstruction within $\omega
$-BPDQ$_p$, with the added observation that standard Gaussian random matrices satisfy these properties with high probability. Moreover, we establish a relationship between RIP$_{p,q}$ and $\omega$-RNSP$_{p,q}$ that RIP$_{p,q}$ implies $\omega$-RNSP$_{p,q}$. Additionally, numerical experiments confirm that the incorporation of weights and the non-Gaussian constraint results in improved reconstruction quality.
\end{abstract}


\section{Introduction}
The fundamental aim of compressive sensing lies in the reconstruction of an unknown signal from the undetermined linear system\cite{donoho2006}\cite{candesorigin}\cite{candesorigin2}. In noisy case, the measurements can be represented as $y=\Phi x+n$, where $x\in\mathbb{R}^N$ denotes the eager unknown signal, $\Phi\in\mathbb{R}^{m\times N}$ is the sensing matrix, $n\in\mathbb{R}^m$ represents the noise and $m<N$.

In practice, prior support information may be available for the signal. 
For instance, in time-series signal processing, the support can be estimated using information from previous time instants\cite{vaswani2010}\cite{friedlander2011}\cite{peng2014}\cite{yu2013}.
And it may be possible to estimate the support of the signal or its largest coefficients\cite{CHEN2019}\cite{ge2021weighted}.
Let $\widetilde{T}\subset[N]:=\{1,2,\ldots,N\}$ denote the support with prior information, and the weight $\omega=[\omega_1,\omega_2,\ldots,\omega_N]\in[0,1]^N$ is introduced:
\begin{equation}
\label{w}
    \omega_i=
\begin{cases}
    \theta \text{ if }i\in\widetilde{T}\\
    1 \text{ if }i\notin\widetilde{T}
\end{cases},
\end{equation}
where constant $\theta\in(0,1)$\cite{zhou2013}\cite{ge2021}.

In previous studies of the $\ell_1$ minimization, if the measurements have a noise level, denoted $n$, it is often assumed that the noise is Gaussian white noise, i.e., noise that follows a Gaussian distribution and is uncorrelated, e.g., \cite{candes2006IEEE}.
At this point, due to the characteristics of the normal distribution, the constraint on the noise is in the form of $\ell_2$ norm, i.e., $\|n\|_2\leqslant\epsilon,\,\epsilon>0$. However, in practical applications, the distribution of measurement value errors is not necessarily normal.

In this work, we consider a particular  model of sensing. A specific quantized model will be introduced. The measurement error after quantization generally follows a uniform distribution instead of a Gaussian distribution.

Quantization is crucial for converting analog signals to digital format. It is mapping input values from a large set to output the nearest approximating values in a smaller set\cite{Gersho1991}.
One of the most common uniform quantizers is the mid-riser quantizer\cite{Gersho1991}\cite{jac2011}\cite{panvan2017}, which will be discussed in the remainder of this paper:

\begin{equation}
\label{quantized}
    y_{\mathrm{q}}=Q_\alpha[\Phi x]=\Phi x+n,
\end{equation}
where $y_{\mathrm{q}} \in\left(\alpha \mathbb{Z}+\frac{\alpha}{2}\right)^m$ is the quantized measurement vector, $\left(Q_\alpha[\cdot]\right)_i=\alpha\left\lfloor(\cdot)_i / \alpha\right\rfloor+\frac{\alpha}{2}$ is the uniform mid-riser quantizer in $\mathbb{R}^m$ of bin width $\alpha$, and $n := Q_\alpha[\Phi x]-\Phi x$ is the quantization distortion.

To a first approximation, the quantization error for each measurement behaves like a uniformly distributed random variable on $[-\frac{\alpha}{2},\frac{\alpha}{2}]$\cite{panvan2017}.

To obtain a more accurate estimate of the error, the form of the noise constraint should be changed.
In \cite{jac2011}, a model BPDQ$_p$ that changes the norm of the noise constraint is proposed:
$$
    \min_{z\in\mathbb{C}^n}\|z\|_{1}\mbox{ subject to }\|\Phi z-y\|_{p}\le\epsilon
$$

The difference from the original $\ell_1$ minimization problem is the replacement of the noise constraint with the $\ell_p$ norm.
In conjunction with the previous weighted form of the $\ell_1$ minimization problem, the above BPDQ$_p$ model can be extended to construct the $\omega$-BPDQ$_p$ model. 
In the problem of reconstructing a sparse vector $x$ by known measurements $y=\Phi x+n$ with $\ell_p$ norm noise constraints, i.e., $\|n\|_p\leqslant\epsilon,\,\epsilon>0$. The $\omega$-BPDQ$_p$ model can be expressed as:
\begin{equation}
\label{wbpdqp}
    \min_{x\in\mathbb{C}^n}\|x\|_{\omega,1}\mbox{ subject to }\|\Phi x-y\|_{p}\le\epsilon.
\end{equation}
where $\|x\|_{\omega,1}:=\sum_i|x_i|\omega_i$.

The restricted isometry property and the null space property for this model have extended versions that ensure the reconstruction, denoted as RIP$_{p,q}$(see in Definition \ref{wrippq}) and $\omega$-RNSP$_{p,q}$( see in Definition \ref{robustnsp}) respectively. In Section \ref{sec:wrip}, after introducing the definitions of the two properties, we respectively show that they are capable of ensuring stable and robust reconstruction through the $\omega$-BPDQ$_p$ model, and can get bounds for reconstruction errors that are nearly of the same order of magnitude(see in Theorem \ref{thmripp2} and \ref{thmrnsp}).

To ensure the theory's feasibility, it is essential to ascertain the existence of matrices that satisfy RIP$_{p,q}$ and $\omega$-RNSP$_{p,q}$ properties. Given a sufficiently large number of measurements $m$, standard Gaussian random matrices are expected to satisfy RIP$_{p,q}$ and $\omega$-RNSP$_{p,q}$ properties with high probability(see in Proposition \ref{nspgaussian}).  To prove Proposition \ref{nspgaussian}, a weighted convex hull is defined. However, there is insufficient evidence to demonstrate that the weighted convex ball is equivalent to the unit ball of a certain norm, as has been proven in the unweighted case \cite{dirk2018}.
However, the minimum number of measurements required for RIP$_{p,q}$ is typically larger than that for $\omega$-RNSP$_{p,q}$, suggesting a possible inclusion relationship between them. 
Indeed, RIP$_{p,q}$ of order $2s$ can imply $\omega$-RNSP$_{p,q}$ of order $s$(see in Proposition \ref{rip2nsp}).

In Section \ref{sec:exp}, to provide experimental validation of the theoretical framework, we introduce the proximal operator to solve the convex optimization. We utilized the Douglas–Rachford algorithm as the iterative method for solving the weighted $\ell_1$ minimization problem arising from $\omega$-BPDQ$_p$. In our experimental analysis, we demonstrate the indispensability of weights and $\ell_p$ constraints by employing random vectors.

Let us introduce some notations that will be mentioned in the following paper. The usual $\ell_\infty$ norm in $\mathbb{C}^N$ is denoted by $\|x\|_\infty=\max_{j=1,\ldots,N}|x_j|$. And $\|x\|_0:=\#\{j:x_j\ne 0\}$. The null space of matrix $A$ is denoted by $\ker(A)$. The index set of support of vector $x\in\mathbb{C}^N$ is denoted by $[N]=\{1,2,\ldots,N\}$. The complement set of set $S$ is $S^c$. And $x_S$ is the vector has the same entries as $x$ on $S$, and 0 on $S^c$.
\section{$\omega$-BPDQ$_p$ and Its Reconstruction Guarantee}
\label{sec:wrip}

In this section, the definitions of RIP$_{p,q}$ and $\omega$-RNSP$_{p,q}$ will be introduced. Subsequently, their capability to provide stable and robust error bounds for $\omega$-BPDQ$_p$ will be demonstrated. Specifically, the discussion will cover how standard Gaussian random matrices, when the number of observations is sufficiently large, can satisfy RIP$_{p,q}$ and $\omega$-RNSP$_{p,q}$ with high probability. Additionally, the relationship between RIP$_{p,q}$ and $\omega$-RNSP$_{p,q}$ will be explored.

\subsection{RIP$_{p,q}$}
To provide theoretical support for the reconstruction of this model, the restricted isometry property(RIP$_{p,q}$) for the $\ell_p$ norm is required\cite{dirk2018}\cite{jac2011}. 
\begin{definition}[RIP$_{p,q}$(see Definition 1 in \cite{jac2011})]
    \label{wrippq}
    For given  $p,\,q>0$,  the matrix $\Phi\in \mathbb{R}^{m\times{N}}$ satisfies RIP$_{p,q}$  of order s with the RIP constant $\delta_s>0$ if there exists a constant $\mu_{p,q}>0$ such that
    $$
    \mu_{p,q}(1-\delta_{s})^{1/q}\|x\|_{q}\le\|\Phi x\|_{p}\le\mu_{p,q}(1+\delta_{s})^{1/q}\|x\|_{q}
    $$
    for arbitrary vector with its support set $S$ satisfying $|S|\leqslant s$.
\end{definition}
The map between $\ell_{p}^{m}=(\mathbb{R}^m,\,\|\cdot \|_p)$ and $\ell_{q}^{N}=(\mathbb{R}^N,\,\|\cdot \|_q)$ is similar to being isometric when $\delta_{s}$ is small enough. If $p=q=2$, RIP$_{2,2}$ is equivalent to RIP.

A reconstruction guarantee for the $\omega$-BPDQ$_p$ model is proved below. Briefly, if the sensing matrix $\Phi$ satisfies RIP$_{p,q}$, stable and robust recovery can be guaranteed with a certain probability for sparse or compressible vectors with $\ell_p$ norm noise constraints, where $p,q\in[2,+\infty)$.

To describe the stability, the error of best weighted s-term approximation is involved. For the compressible vectors which are not exactly sparse, stability ensures that the reconstructed vector remains sufficiently close to the original one.

For $s\geqslant1$, the error of best weighted s-term approximation of the vector $\boldsymbol{x}\in\mathbb{C}^N$ is defined as
\begin{equation}
\label{sigmax}    \sigma_s(\boldsymbol{x})_{\omega,1}=\inf_{\|\boldsymbol{z}\|_{\omega,0}\leqslant s}\|\boldsymbol{x-z}\|_{\omega,1}.
\end{equation}

The robustness is represented by the constraint $\|\Phi x-y\|_p\leqslant\epsilon$, characterizes the ability of the reconstruction method to handle noise and errors in the measurements.
\begin{theorem}
    \label{thmripp2}
    For any given weight $\omega=[\omega_1,\ldots,\omega_N]^T\in [\theta,1]^N$, $\theta\in(0,1)$, which is defined as \eqref{w}, if the matrix $\Phi$ satisfies RIP$_{p,q}$ of order $2s$ with $\delta_{2s}$, $s\geqslant2$ and $p,q\in[2,+\infty)$, then the vector $x\in\mathbb{C}^N$ and the solution of \eqref{wbpdqp} satisfy
    \begin{equation*}
        \|\hat{x}-x\|_q\leqslant A\frac{\epsilon}{\mu_{p,q}}+Bs^{1/q-1}\sigma_{s}(x)_{\omega,1}.
    \end{equation*}
    where $\epsilon$ is defined as \eqref{wbpdqp}, $\sigma_s(\boldsymbol{x})_{\omega,1}$ is defined as \eqref{sigmax}, A and B depend on $\delta_{2s}$, $p$, $q$ and $\theta$.
\end{theorem}
Before proving the theorem, it is necessary to show that for any two vectors whose supports are disjoint, some form of inner product associated with the matrix and the vector can be expressed in terms of the norms of the two vectors if the matrices satisfy RIP$_{p,q}$ of the corresponding order.
\begin{lemma}[see Lemma 2 in \cite{jac2011}]
    \label{lemma2}
    If $u,\,v\in\mathbb{R}^N$ such that $\|u\|_{0}=s,\,\|v\|_{0}=s'$, $\operatorname{supp}(u)\cap \operatorname{supp}(v)=\emptyset$ and the matrix $\Phi$ satisfies RIP$_{p,q}$, and $p,q\in[2,+\infty)$ of order $s,\,s',\,s+s'$ with $\delta_s$, $\delta_{s'}$ and $\delta_{s+s'}$ respectively, then 
    \begin{equation}
        |\langle J(\Phi u),\Phi v\rangle|\leqslant \mu_{p,q}^2C_{p,q}\|u\|_q\|v\|_q,
    \end{equation}
    where $(J(u))_i=\|u\|_p^{2-p}|u_i|^{p-1}\operatorname{sign}(u_i),\,C_{p,q}=C_{p,q}(\Phi,\,s,\,s')$.
\end{lemma}
This proof is based on Lemma 2 approach in \cite{jac2011}, with adjustments made on the norm.
\begin{proof}
   According to the definition of $J$, $J(\lambda x)=\lambda J(x)$ for any $\lambda\in\mathbb{R}$ and $x\in\mathbb{R}^m$. Without loss of generality, assume that $\|u\|_q=\|v\|_q=1$.

   In the Banach space $\ell_p^m=(\mathbb{R}^m,\|\cdot\|_p)$ for $p\geqslant2$, there is a dual space $\ell_{p^*}^m=(\mathbb{R}^m,\|\cdot\|_{p^*})$ such that $\frac{1}{p}+\frac{1}{p^*}=1$.

With the definition of $J(x)$, we have
   \begin{align*}
       \langle J(x),x\rangle&=\sum_i\|x\|_p^{2-p}|x_i|^{p}=\|x\|_p^2;\\
       \|J(x)\|_{p^*}^2&=[\sum_i(\|x\|_p^{2-p}|x_i|^{p-1})^{p^*}]^\frac{2}{p^*}\\
       &=\|x\|_p^{2(2-p)}(\sum_i|x_i|^p)^{\frac{2(p-1)}{p}}\\
       &=\|x\|_p^2.
   \end{align*}

   According to the Definition 1 in \cite{DePrima1971}, $J(x)$ is a normalized duality mapping of $\ell_p^m$ to $\ell_{p^*}^m$.
   Then from Theorem 3 in \cite{Bynum_1976}, there is an inequality that:
$$
\|x+y\|_{p}^{2} \leqslant\|x\|_{p}^{2}+2\langle J(x), y\rangle+(p-1)\|y\|_{p}^{2}.
$$

Let $y=-y$, we have
\begin{equation}
    \label{l1}
    2\langle J(x), y\rangle \leqslant\|x\|_{p}^{2}+(p-1)\|y\|_{p}^{2}-\|x-y\|_{p}^{2}.
\end{equation}
Let $x=\Phi u$ and $y=t \Phi v$ for some $t>0$. With RIP$_{p,q}$ of order $s,\,s',\,s+s'$, it implies that 
$$
2t\langle J(\Phi u), \Phi v\rangle\leqslant\|\Phi u\|_{p}^{2}+(p-1)\|t\Phi v\|_{p}^{2}-\|\Phi(u-t v)\|_{p}^{2}.
$$

With RIP$_{p,q}$ of order $s,\,s',\,s+s'$, it implies that 
\begin{equation}
    \label{ie1}
    \begin{split}
        2 \mu_{p, q}^{-2} t\langle J(\Phi u), \Phi v\rangle &\leqslant\left(1+\delta_{s}\right)^{2/q}
+(p-1)\left(1+\delta_{s^{\prime}}\right)^{2/q} t^{2}\\
&\quad-\left(1-\delta_{s+s^{\prime}}\right)^{2/q}\left(1+t^{2}\right).
    \end{split}
\end{equation}

If let $t=-t$, we have
\begin{equation}
    \label{ie2}
    \begin{split}
         -2 \mu_{p, q}^{-2} t\langle J(\Phi u), \Phi v\rangle &\leqslant\left(1+\delta_{s}\right)^{2/q}
+(p-1)\left(1+\delta_{s^{\prime}}\right)^{2/q} t^{2}\\
&\quad-\left(1-\delta_{s+s^{\prime}}\right)^{2/q}\left(1+t^{2}\right).
    \end{split}
\end{equation}

Combine \eqref{ie1} and \eqref{ie2}, the left hand side can be written as the absolute value:
\begin{equation*}
    \begin{split}
        2 \mu_{p, q}^{-2} t|\langle J(\Phi u), \Phi v\rangle |&\leqslant\left(1+\delta_{s}\right)^{2/q}
+(p-1)\left(1+\delta_{s^{\prime}}\right)^{2/q} t^{2}\\
&\quad-\left(1-\delta_{s+s^{\prime}}\right)^{2/q}\left(1+t^{2}\right).
    \end{split}
\end{equation*}

Let 
\begin{equation*}
    \begin{split}
        f(t)&=\frac{1}{2t}(
\left(1+\delta_{s}\right)^{2/q}
+(p-1)\left(1+\delta_{s^{\prime}}\right)^{2/q} t^{2}\\
&\quad-\left(1-\delta_{s+s^{\prime}}\right)^{2/q}\left(1+t^{2}\right)).
    \end{split}
\end{equation*}

Then find its minimum with derivative,
\begin{equation*}
    \begin{split}
        \min_{t>0}f(t)&=\{[(1+\delta_s)^{2/q}-(1+\delta_{s+s'})^{2/q}]\\
       &\quad [(p-1)(1+\delta_{s'})^{2/q}-(1+\delta_{s+s'})^{2/q}]\}^{1/2},
    \end{split}
\end{equation*}

i.e.
\begin{equation}
    \begin{split}
    \label{l2}
        \mu_{p, q}^{-2} |\langle J(\Phi u), \Phi v\rangle|&\leqslant
[(1+\delta_s)^{2/q}-(1+\delta_{s+s'})^{2/q}]^{1/2}\\
&[(p-1)(1+\delta_{s'})^{2/q}-(1+\delta_{s+s'})^{2/q}]^{1/2},
    \end{split}
\end{equation}

Then define $C_{p,q}=[(1+\delta_s)^{2/q}-(1+\delta_{s+s'})^{2/q}]^{1/2}[(p-1)(1+\delta_{s'})^{2/q}-(1+\delta_{s+s'})^{2/q}]^{1/2}$, it implies that $|\langle J(\Phi u),\Phi v\rangle|\leqslant \mu_{p,q}^2C_{p,q}\|u\|_q\|v\|_q$.
\end{proof}
\begin{proof}[proof of \ref{thmripp2}]
Assume $\hat{x}=x+h$, then $h$ is the noise. While $x$ can be any vector in $\mathbb{C}^N$, $h$ is also an arbitrary vector.
There exists an index set $S =: S_0$ with $|S|= s$ such that $\sigma_{s}(x)_{\omega,1}=\|x-x_{S}\|_{\omega,1}=\|x_{S^{c}}\|_{\omega,1}$. 
We make a nonincreasing rearrangement of the entries in set $S^{c}$ based on the values of $h_{i}\omega_{i}^{-1}$ with $i\in S^c$, and partition $S_0^c$ as $\{S_l:1\leqslant l\leqslant\lceil\frac{N-s}{s}\rceil \}$ with $|S_l|=s$ for $1\leqslant l\leqslant\lfloor\frac{N-s}{s}\rfloor$ and $|S_l|\leqslant s$ for $\lceil\frac{N-s}{s}\rceil$.
So that $|h_j|\omega_j^{-1}\leqslant|h_k|\omega_k^{-1}$ for all $j\in S_l$, $k\in S_{l-1}$, $l\geqslant2$. We denote that $S_{01}=S_0\cup S_1$.

     According to the definition of $\ell_{\omega,1}$ norm and $q\geqslant2$, 
     \begin{equation}
     \label{w2q}
     \begin{split}
         \|h_{S}\|_{\omega,1}&=\sum_{j\in S}|h_{j}|\omega_{j}\\
         &\leqslant(\sum_{j\in S}|h_{j}|^q)^{1/q}(\sum_{j\in S}\omega_{j}^{\frac{q}{q-1}})^{1-\frac{1}{q}}\\
         &\leqslant s^{1-\frac{1}{q}}\|h_{S}\|_{q}.
            \end{split}
     \end{equation}
  
   According to the definition of the matrix $J$ in Lemma \ref{lemma2},
   $$
    \|\Phi h_{S_{01}}\|_p^2=\langle J(\Phi h_{S_{01}}),\Phi h\rangle-\langle J(\Phi h_{S_{01}}),\Phi\left(\sum_{l\geqslant2}h_{S_l}\right)\rangle.
    $$
    Since $h=\hat{x}-x$,
    \begin{equation}
        \label{phih}
        \|\Phi h\|_p\leqslant \|\Phi x-y\|_p+\|\Phi\hat{x}-y\|_p\leqslant 2\epsilon.
    \end{equation}
     By the H$\ddot{o}$lder inequality and $\omega$-$RIP_{p,q}$ on $S_{01}$,
 \begin{align*}
    |\langle J(\Phi h_{S_{01}}),\Phi h\rangle|&\leqslant\|J(\Phi h_{S_{01}})\|_{\frac{p}{p-1}}\|\Phi h\|_p\\
    &=\|\Phi h_{S_{01}}\|_p\|\Phi h\|_p\\
    &\leqslant 2\epsilon\|\Phi h_{S_{01}}\|_p\\
    &\leqslant 2\epsilon\mu_{p,q}(1+\delta_{2s})^{1/q}\|h_{S_{01}}\|_q.
  \end{align*}
 When $l\geqslant2$, $S_{01}$ and $S_l$ are disjoint. By Lemma \ref{lemma2}, if the matrix $\Phi$ satisfies $\omega$-$RIP_{p,q}$ of order $2s$, then
         $$
|\langle J(\Phi h_{S_{01}}),\Phi h_{S_l}\rangle|\leqslant \mu_{p,q}^2C_{p,q}\|h_{S_{01}}\|_q\|h_{S_l}\|_q.
    $$
    Combining the equations above, we have
     \begin{align*}
    \mu_{p,q}^2(1-\delta_{2s})^{2/q}\|h_{S_{01}}\|_q^{2}&\leqslant\|\Phi h_{S_{01}}\|_p^2\\
    &\leqslant 2\epsilon\mu_{p,q}(1+\delta_{2s})^{1/q}\|h_{S_{01}}\|_q\\
&\quad+\mu_{p,2}^2C_{p,q}\|h_{S_{01}}\|_q\sum_{l\geqslant2}\|h_{S_l}\|_q.
  \end{align*}
    Divide both sides by$\|h_{S_{01}}\|_q$,
    \begin{equation}
    \label{l2tl2}
        \|h_{S_{01}}\|_q\leqslant\frac{2(1+\delta_{2s})^{1/q}}{\mu_{p,q}(1-\delta_{2s})^{2/q}}\epsilon+\frac{C_{p,q}}{(1-\delta_{2s})^{2/q}}\sum_{l\geqslant2}\|h_{S_l}\|_q.
    \end{equation}
     
     For arbitrary $k\in S_l$, assume the parameters $\alpha_k=(\sum_{j\in S_l}\omega_j^2)^{-1}\omega_k^2\leqslant(s\theta^2)^{-1}\omega_k^2$, obviously that $\sum_{k\in S_l}\alpha_{k}=1$.

Then for $ \forall j\in S_l,\,l\geqslant 2,$ 

$$
    |h_j|\omega_j^{-1}\leqslant \sum_{k\in S_{l-1}}\alpha_k|h_k|\omega_k^{-1}\leqslant (s\theta^2)^{-1}\sum_{k\in S_{l-1}}|h_k|\omega_k.
    $$
    By Cauchy-Schwartz inequality and $q\geqslant2$, for any $l\geqslant 2$, 
    \begin{equation*}
    \begin{split}
         \|h_{S_l}\|_q&=(\sum_{j\in S_l}(|h_j|\omega_j^{-1})^q\omega_j^q)^{\frac{1}{q}}\\
         &\leqslant\frac{s^{1/q}}{s\theta^2}\|h_{S_{l-1}}\|_{\omega,1}\\
         &\leqslant \frac{1}{\theta^2}s^{1/q-1}\|h_{S_{l-1}}\|_{\omega,1},
    \end{split}
\end{equation*}

\begin{equation}
\label{ll2q}
\begin{split}
     \|h_{S_{01}^c}\|_q&\leqslant \sum_{l\geqslant 2}\|h_{S_{l}}\|_q\\
     &\leqslant\frac{1}{\theta^2}s^{1/q-1}(\|h_{S_1}\|_{\omega,1}+\|h_{S_2}\|_{\omega,1}+\ldots)\\
     &\leqslant\frac{1}{\theta^2}s^{1/q-1}\|h_{S^c}\|_{\omega,1}.
     \end{split}
\end{equation}

  Substitute $\|h_{S_l}\|_q$ in \eqref{l2tl2} with the above estimate, we have  
    \begin{equation}
    \label{hs2}
        \begin{split}
\|h_S\|_q&\leqslant\|h_{S_{01}}\|_q\\
&\leqslant\frac{2(1+\delta_{2s})^{1/q}}{\mu_{p,q}(1-\delta_{2s})^{2/q}}\epsilon+\frac{C_{p,q}}{\theta^2(1-\delta_{2s})^{2/q}}s^{1/q-1}\|h_{S_{01}^c}\|_{\omega,1}.
\end{split}
    \end{equation}

  By the triangle inequality of the norms,
    \begin{align*}
      \|x\|_{\omega,1}+\|(x-\hat{x})_{S^{c}}\|_{\omega,1}&\leqslant 2\|x_{S^{c}}\|_{\omega,1}+\|\hat{x}_{S^{c}}\|_{\omega,1}+\|x_{S}\|_{\omega,1}\\
      &\leqslant 2\sigma_{s}(x)_{\omega,1}\\
      &\quad+\|(x-\hat{x})_{S}\|_{\omega,1}+\|\hat{x}\|_{\omega,1}.
    \end{align*}
    Since $h=\hat{x}-x$ and $\|\hat{x}\|_{\omega,1}\leqslant\|x\|_{\omega,1}$, 
    \begin{equation}
    \label{hsc}
         \begin{split}
\|h_{S_{01}^c}\|_{\omega,1}&\leqslant\|h_{S^{c}}\|_{\omega,1}\\
&\leqslant 2\sigma_{s}(x)_{\omega,1}+\|\hat{x}\|_{\omega,1}-\|x\|_{\omega,1}+\|h_{S}\|_{\omega,1}\\
     & \leqslant 2\sigma_{s}(x)_{\omega,1}+\|h_{S}\|_{\omega,1}.
      \end{split}
    \end{equation}
   
    Then 
    \begin{equation*}
    \begin{split}
        \|h_{S}\|_q&\leqslant\frac{2(1+\delta_{2s})^{1/q}}{\mu_{p,q}(1-\delta_{2s})^{2/q}}\epsilon\\
        &\quad+\frac{C_{p,q}}{\theta^2(1-\delta_{2s})^{2/q}}s^{1/q-1}(2\sigma_{s}(x)_{\omega,1}+s^{(1-1/q)}\|h_{S}\|_q).
         \end{split}
\end{equation*}
       
After simplifying the inequation above, we get         
\begin{equation}
\label{hsq}
\begin{split}
      \|h_{S}\|_q&\leqslant\frac{2\theta^2(1+\delta_{2s})^{1/q}}{\mu_{p,q}[C_{p,q}-\theta^2(1-\delta_{2s})^{2/q}]}\epsilon\\
      &\quad+\frac{C_{p,q}\theta^2s^{1/q-1}}{[C_{p,q}-\theta^2(1-\delta_{2s})^{2/q}]}\sigma_{s}(x)_{\omega,1}.
      \end{split}
\end{equation}

With \eqref{l2tl2}, \eqref{ll2q}, \eqref{hsc}, \eqref{hsq} and \eqref{w2q}, the error bound in $\ell_q$ norm is 
\begin{align*}
    &\|h\|_q\leqslant\|h_{S_{01}}\|_q+\|h_{S_{01}^c}\|_q\\
    &\leqslant\frac{2(1+\delta_{2s})^{1/q}}{\mu_{p,q}(1-\delta_{2s})^{2/q}}\epsilon+\frac{C_{p,q}}{(1-\delta_{2s})^{2/q}}\sum_{l\geqslant2}\|h_{S_l}\|_q+\|h_{S_{01}^c}\|_q\\
    &\leqslant\frac{2(1+\delta_{2s})^{1/q}}{\mu_{p,q}(1-\delta_{2s})^{2/q}}\epsilon+[\frac{C_{p,q}}{(1-\delta_{2s})^{2/q}}+1]s^{1/q-1}\sum_{l\geqslant2}\|h_{S_l}\|_q\\
    &\leqslant\frac{2(1+\delta_{2s})^{1/q}}{\mu_{p,q}(1-\delta_{2s})^{2/q}}\epsilon\\
    &\quad+[\frac{C_{p,q}}{(1-\delta_{2s})^{2/q}}+1]s^{1/q-1}\frac{1}{\theta^2}s^{1/q-1}\|h_{S^c}\|_{\omega,1}\\
    &\leqslant\frac{2(1+\delta_{2s})^{1/q}}{\mu_{p,q}(1-\delta_{2s})^{2/q}}\epsilon\\
    &\quad+[\frac{C_{p,q}}{(1-\delta_{2s})^{2/q}}+1]\frac{1}{\theta^2}s^{1/q-1}(2\sigma_{s}(x)_{\omega,1}+s^{(1-1/q)}\|h_{S}\|_q)\\
    &\leqslant A\frac{\epsilon}{\mu_{p,q}}+Bs^{1/q-1}\sigma_{s}(x)_{\omega,1},
\end{align*}
where $A$ and $B$ are constants depending on $C_{p,q}$, $\delta_{2s}$ and $\theta$.

\end{proof}
\begin{remark}
 The constants $C_{p,q}$ and $B$ in Theorem \ref{thmripp2} and Lemma \ref{lemma2} increase with $p$.
In the $\omega$-BPDQ$_\infty$ model, if $p=\infty$, a valid error bound estimate cannot be obtained with infinitely large constants.
With this process of estimating error bounds, the signals are not guaranteed to be reconstructed.
\end{remark}

\subsection{$\omega$-RNSP$_{p,q}$}

$\omega$-RNSP$_{p,q}$ is defined by substituting the $\ell_1$ norm and $\ell_2$ constraint in $\ell_q$-RNSP with $\ell_{\omega,1}$ and $\ell_p$ norm, respectively. These adjustments are analogous to transitioning from $\ell_1$ minimization to $\omega$-BPDQ$_{p}$. 
\begin{definition}[$\omega$-RNSP$_{p,q}$]
        \label{robustnsp}
    For given weights $\omega$, constants $\rho\in(0.1),\,\gamma>0$ and $p\geqslant1$, the matrix $A\in \mathbb{C}^{m\times{N}}$ is said to satisfy the weighted robust null space property($\omega$-RNSP$_{p,q}$) of order s if it satisfies:
    \begin{equation}
        \|\upsilon_{S}\|_{q}\le\frac{\rho}{s^{1-\frac{1}{q}}}\|\upsilon_{S^{c}}\|_{\omega,1}+\gamma\|A\upsilon \|_{p}
    \end{equation}
    for all $\upsilon\in\mathbb{C}^{N}$and all $S\subset[N]$ with $|S|\le s$.
\end{definition}

Similar to RIP$_{p,q}$, $\omega$-RNSP$_{p,q}$ can also provide stable and robust error bound.
\begin{theorem}
    \label{thmrnsp}
    For any given weight $\omega=[\omega_1,\ldots,\omega_N]^T\in [0,1]^N$, which is defined as \eqref{w}, if the matrix $A\in \mathbb{C}^{m\times{N}}$ satisfies $\omega$-RNSP$_{p,q}$ of order s with $p,q\in[1,+\infty]$, $\rho\in(0.1)$, $\gamma>0$ and $s\geqslant2$, then for any $1\leqslant r\leqslant q$, the vector $x\in\mathbb{C}^N$ and the solution of \eqref{wbpdqp} satisfy
   \begin{equation}
        \|\hat{x}-x\|_r\leqslant As^{1/r-1/q}\epsilon+Bs^{1/r-1}\sigma_{s}(x)_{\omega,1},
    \end{equation}
     where $\epsilon$ is defined as \eqref{wbpdqp}, $\sigma_s(\boldsymbol{x})_{\omega,1}$ is defined as \eqref{sigmax}, A and B depend on $\rho$, $\gamma$ and $\theta$.
\end{theorem}
\begin{proof}
    For any $1\leqslant r\leqslant q$, $S\subset[N]$ and $|S|\leqslant s$, $\|v_S\|_r\leqslant s^{1/r-1/q}\|v_S\|_q$ for all $v\in \mathbb{C}^N$.\\
    So for all $v\in \mathbb{C}^N$, we have
    \begin{equation}
    \label{nsp_r}
        \|v_S\|_r\leqslant\frac{\rho}{s^{1-1/r}}\|v_{S^c}\|_{\omega,1}+\gamma s^{1/r-1/q}\|\Phi v\|_p.
    \end{equation}

Assume $\hat{x}=x+h$, then $h$ is the noise. While $x$ can be any vector in $\mathbb{C}^N$, $h$ is also an arbitrary vector.
There exists an index set $S =: S_0$ with $|S|\leqslant s$ such that $\sigma_{s}(x)_{\omega,1}=\|x-x_{S}\|_{\omega,1}=\|x_{S^{c}}\|_{\omega,1}$. 
We make a nonincreasing rearrangement of the entries in set $S^{c}$ based on the values of $h_{i}\omega_{i}^{-1}$ with $i\in S^c$, and partition $S_0^c$ as $\{S_l:1\leqslant l\leqslant\lceil\frac{N-s}{s}\rceil \}$ with $|S_l|=s$ for $1\leqslant l\leqslant\lfloor\frac{N-s}{s}\rfloor$ and $|S_l|\leqslant s$ for $\lceil\frac{N-s}{s}\rceil$.
So that $|h_j|\omega_j^{-1}\leqslant|h_k|\omega_k^{-1}$ for all $j\in S_l$, $k\in S_{l-1}$, $l\geqslant2$. We denote that $S_{01}=S_0\cup S_1$.

By the triangle inequality of the norms,
    \begin{align*}
      \|x\|_{\omega,1}+\|(x-\hat{x})_{S^{c}}\|_{\omega,1}&\leqslant 2\|x_{S^{c}}\|_{\omega,1}+\|\hat{x}_{S^{c}}\|_{\omega,1}+\|x_{S}\|_{\omega,1}\\
      &\leqslant 2\sigma_{s}(x)_{\omega,1}+\|(x-\hat{x})_{S}\|_{\omega,1}+\|\hat{x}\|_{\omega,1}.
    \end{align*}
    Since $h=\hat{x}-x$ and $\|\hat{x}\|_{\omega,1}\leqslant\|x\|_{\omega,1}$, 
 \begin{align*}
         \|h_{S^{c}}\|_{\omega,1}&\leqslant 2\sigma_{s}(x)_{\omega,1}+\|\hat{x}\|_{\omega,1}-\|x\|_{\omega,1}+\|h_{S}\|_{\omega,1}\\
      &\leqslant 2\sigma_{s}(x)_{\omega,1}+\|h_{S}\|_{\omega,1}.
    \end{align*}

    Suppose first that $q<\infty$.
   According to the definition of $\ell_{\omega,1}$ norm, and $1\leqslant r\leqslant q$, 
    \begin{align*}
        \|h_{S}\|_{\omega,1}&=\sum_{j\in S}|h_{j}|\omega_{j}\leqslant(\sum_{j\in S}|h_{j}|^r)^{1/r}(\sum_{j\in S}\omega_{j}^{\frac{r}{r-1}})^{1-\frac{1}{r}}\leqslant s^{1-\frac{1}{r}}\|h_{S}\|_{r}\\
        &\leqslant s^{1-\frac{1}{r}}(\frac{\rho}{s^{1-1/r}}\|h_{S^c}\|_{\omega,1}+\gamma s^{1/r-1/q}\|\Phi h\|_p)\\
        &\leqslant\rho\|h_{S^c}\|_{\omega,1}+\gamma s^{1-\frac{1}{q}}\|\Phi h\|_p.
    \end{align*}

    So we have
    \begin{align*}
        \|h_{S^{c}}\|_{\omega,1}&\leqslant 2\sigma_{s}(x)_{\omega,1}+\|h_{S}\|_{\omega,1}\\
        &\leqslant2\sigma_{s}(x)_{\omega,1}+\rho\|h_{S^c}\|_{\omega,1}+\gamma s^{1-\frac{1}{q}}\|\Phi h\|_p\\
        &\leqslant\frac{1}{1-\rho}(2\sigma_{s}(x)_{\omega,1}+\gamma s^{1-\frac{1}{q}}\|\Phi h\|_p).
    \end{align*}

    For arbitrary $k\in S_l$, assume the parameters $\alpha_k=(\sum_{j\in S_l}\omega_j^2)^{-1}\omega_k^2\leqslant(\theta^2s)^{-1}\omega_k^2$, obviously that $\sum_{k\in S_l}\alpha_{k}=1$.\\
    Then  $\forall j\in S_l$,\,$l\geqslant 2$,
$$
     |h_j|\omega_j^{-1}\leqslant \sum_{k\in S_{l-1}}\alpha_k|h_k|\omega_k^{-1}\leqslant (\theta^2s)^{-1}\sum_{k\in S_{l-1}}|h_k|\omega_k.
    $$
    
    By Cauchy-Schwartz inequality and $1\leqslant r\leqslant q$, for any $l\geqslant 2$, 
    \begin{align*}
        \|h_{S_l}\|_r=(\sum_{j\in S_l}(|h_j|\omega_j^{-1})^r\omega_j^r)^{\frac{1}{r}}
        &\leqslant\frac{s^{1/r}}{\theta^2s}\|h_{S_{l-1}}\|_{\omega,1}\\
        &\leqslant \frac{1}{\theta^2}s^{1/r-1}\|h_{S_{l-1}}\|_{\omega,1},
    \end{align*}
  \begin{align*}
       \|h_{S_{01}^c}\|_r&\leqslant \sum_{l\geqslant 2}\|h_{S_{l}}\|_r\\
       &\leqslant\frac{1}{\theta^2}s^{1/r-1}(\|h_{S_1}\|_{\omega,1}+\|h_{S_2}\|_{\omega,1}+\ldots)\\
       &\leqslant\frac{1}{\theta^2}s^{1/r-1}\|h_{S^c}\|_{\omega,1}.
  \end{align*}

Then we have
\begin{align*}
    \|h\|_r&\leqslant \|h_{S_{01}}\|_r+\|h_{S_{01}^c}\|_r\\
    &\leqslant \|h_{S}\|_r+\|h_{S_1}\|_r+\frac{1}{\theta^2}s^{1/r-1}\|h_{S^c}\|_{\omega,1}\\
    &\leqslant \frac{\rho}{s^{1-1/r}}\|h_{S^c}\|_{\omega,1}+\frac{\rho}{s^{1-1/r}}\|h_{S_1^c}\|_{\omega,1}\\
    &\quad+2\gamma s^{1/r-1/q}\|\Phi h\|_p+\frac{1}{\theta^2}s^{1/r-1}\|h_{S^c}\|_{\omega,1}\\
    &\leqslant \frac{\rho}{s^{1-1/r}}\|h_{S^c}\|_{\omega,1}+\frac{\rho}{s^{1-1/r}}\|h\|_{\omega,1}\\
    &\quad+2\gamma s^{1/r-1/q}\|\Phi h\|_p+\frac{1}{\theta^2}s^{1/r-1}\|h_{S^c}\|_{\omega,1}\\
    &\leqslant (\rho+\frac{1}{\theta^2})s^{1/r-1}\|h_{S^c}\|_{\omega,1}\\
    &\quad+\rho s^{1/r-1/q}[(\rho+1)\|h_{S^c}\|_{\omega,1}\\
    &\quad+\gamma s^{1-1/q}\|\Phi h\|_p]+2\gamma s^{1/r-1/q}\|\Phi h\|_p\\
    &\leqslant\frac{\rho^2+2\rho+\frac{1}{\theta^2}}{\rho-1}s^{1/r-1}(2\sigma_{s}(x)_{\omega,1}+\gamma s^{1-\frac{1}{q}}\|\Phi h\|_p)\\
    &\quad+(\rho+2)\gamma s^{1/r-1/q}\|\Phi h\|_p\\
    &\leqslant\frac{\rho^2+2\rho+\frac{1}{\theta^2}}{\rho-1}2s^{1/r-1}\sigma_{s}(x)_{\omega,1}\\
    &\quad+\frac{2\rho^2+3\rho+\frac{1}{\theta^2}-2}{\rho-1}\gamma s^{1/r-1/q}\|\Phi h\|_p.
    \end{align*}
While $\|\Phi h\|_p\leqslant \|\Phi x-y\|_p+\|\Phi\hat{x}-y\|_p\leqslant2\epsilon$, we have
$$
    \|h\|_r\leqslant As^{1/r-1/q}\epsilon+Bs^{1/r-1}\sigma_{s}(x)_{\omega,1},
$$
where $A$ and $B$ depend on $\rho$, $\gamma$ and $\theta$.

Finally, let $r=q=\infty$.
 \begin{align*}
        \|h_{S}\|_{\omega,1}&=\sum_{j\in S}|h_{j}|\omega_{j}\leqslant s\|h_{S}\|_{\infty}\\
        &\leqslant s(\frac{\rho}{s}\|h_{S^c}\|_{\omega,1}+\gamma \|\Phi h\|_p)\\
        &\leqslant\rho\|h_{S^c}\|_{\omega,1}+\gamma s\|\Phi h\|_p.
    \end{align*}

     So we have
    \begin{align*}
        \|h_{S^{c}}\|_{\omega,1}&\leqslant 2\sigma_{s}(x)_{\omega,1}+\|h_{S}\|_{\omega,1}\\
        &\leqslant2\sigma_{s}(x)_{\omega,1}+\rho\|h_{S^c}\|_{\omega,1}+\gamma s\|\Phi h\|_p\\
        &\leqslant\frac{1}{1-\rho}(2\sigma_{s}(x)_{\omega,1}+\gamma s\|\Phi h\|_p).
    \end{align*}

For arbitrary $k\in S_l$, assume the parameters $\alpha_k=(\sum_{j\in S_l}\omega_j^2)^{-1}\omega_k^2\leqslant(\theta^2s)^{-1}\omega_k^2$, obviously that $\sum_{k\in S_l}\alpha_{k}=1$.\\
$$
      \forall j\in S_l,\,l\geqslant 2,\,|h_j|\omega_j^{-1}\leqslant \sum_{k\in S_{l-1}}\alpha_k|h_k|\omega_k^{-1}\leqslant (\theta^2s)^{-1}\sum_{k\in S_{l-1}}|h_k|\omega_k.
    $$

   While $\omega_j\leqslant1$, for any $l\geqslant 2$, 
    $$
      \|h_{S_l}\|_\infty=\max_{j\in S_l}(|h_j|\omega_j^{-1})\omega_j
      \leqslant\max_{j\in S_l}(|h_j|\omega_j^{-1})
      \leqslant\frac{1}{\theta^2s}\|h_{S_{l-1}}\|_{\omega,1},
    $$

\begin{align*}
    \|h_{S_{01}^c}\|_\infty&\leqslant \sum_{l\geqslant 2}\|h_{S_{l}}\|_\infty\\
   &\leqslant\frac{1}{\theta^2s}(\|h_{S_1}\|_{\omega,1}+\|h_{S_2}\|_{\omega,1}+\ldots)\\
    &\leqslant\frac{1}{\theta^2s}\|h_{S^c}\|_{\omega,1}.
\end{align*}

Combine the inequalities above, it implies that
\begin{align*}
    \|h\|_\infty&\leqslant \|h_{S_{01}}\|_\infty+\|h_{S_{01}^c}\|_\infty\\
    &\leqslant \|h_{S}\|_\infty+\|h_{S_1}\|_\infty+\frac{1}{\theta^2s}\|h_{S^c}\|_{\omega,1}\\
    &\leqslant \frac{\rho}{s}\|h_{S^c}\|_{\omega,1}+\frac{\rho}{s}\|h_{S_1^c}\|_{\omega,1}+2\gamma\|\Phi h\|_p+\frac{1}{\theta^2s}\|h_{S^c}\|_{\omega,1}\\
    &\leqslant \frac{\rho}{s}\|h_{S^c}\|_{\omega,1}+\frac{\rho}{s}\|h\|_{\omega,1}+2\gamma\|\Phi h\|_p+\frac{1}{\theta^2s}\|h_{S^c}\|_{\omega,1}\\
    &\leqslant (\rho+\frac{1}{\theta^2})s^{-1}\|h_{S^c}\|_{\omega,1}\\
    &\quad+\frac{\rho}{s} [(\rho+1)\|h_{S^c}\|_{\omega,1}+\gamma s\|\Phi h\|_p]+2\gamma \|\Phi h\|_p\\
    &\leqslant\frac{\rho^2+2\rho+\frac{1}{\theta^2}}{\rho-1}s^{-1}(2\sigma_{s}(x)_{\omega,1}+\gamma \|\Phi h\|_p)\\
    &\quad+(\rho+2)\gamma s\|\Phi h\|_p\\
    &\leqslant\frac{\rho^2+2\rho+\frac{1}{\theta^2}}{\rho-1}2\frac{\sigma_{s}(x)_{\omega,1}}{s}+\frac{2\rho^2+3\rho+\frac{1}{\theta^2}-2}{\rho-1}\gamma \|\Phi h\|_p\\
    &\leqslant A\epsilon+B\frac{\sigma_{s}(x)_{\omega,1}}{s},
\end{align*}
where $A$ and $B$ depend on $\rho$, $\gamma$ and $\theta$. 
\end{proof}
\subsection{With Gaussian matrices}
To maintain accessibility in our discussion, it's worth noting that Standard Gaussian Random matrices are among the most commonly chosen types of sensing matrices in practical applications. In this section, Standard Gaussian Random matrices are exhibited to satisfy RIP$_{p,q}$ for $q=2$, and $\omega$-RNSP$_{p,q}$ with high probability when the number of observations $m$ is sufficiently large.

A Standard Gaussian Random matrix $\Phi\in \mathbb{R}^{m \times N}$ can be described as follows:
\[ \Phi = \begin{bmatrix}
    \Phi_{11} & \Phi_{12} & \cdots & \Phi_{1N} \\
    \Phi_{21} & \Phi_{22} & \cdots & \Phi_{2N} \\
    \vdots & \vdots & \ddots & \vdots \\
    \Phi_{m1} & \Phi_{m2} & \cdots & \Phi_{mN}
\end{bmatrix}_{m\times N} \]

Where each \( \Phi_{ij}\) is sampled independently from the standard normal distribution, i.e. \(\Phi_{ij}\sim \mathcal{N}(0,1)\).

In \cite{jac2011}, the standard Gaussian random matrix $\Phi \in \mathbb{R}^{m \times N}$ is considered to satisfied RIP$_{p,2}$  of order $s$ and RIP constant $\delta$ with a probability exceeding $1-\eta$ when the number of rows $m$ satisfies certain conditions, specifically, $m \geqslant(p-1) 2^{p+1}$ for $2 \leqslant p<\infty$ and $m \geqslant 0$ for $p=\infty$, there exists a positive constant $c$ such that $m\geqslant c \delta^{-2}\left(s \log \left[e \frac{N}{s}\left(1+12 \delta^{-1}\right)\right]+\log \frac{2}{\eta}\right)^{p/2}$ for $2\leqslant p<\infty$ and $m\gtrsim \exp{c \delta^{-2}\left(s \log \left[e \frac{N}{s}\left(1+12 \delta^{-1}\right)\right]+\log \frac{2}{\eta}\right)}$ for $p=\infty$.

\begin{proposition}
     \label{nspgaussian}
     Let $\Phi$ be an $m \times N$ standard Gaussian matrix. For the given $\omega\in[\theta,1]^N$, $\theta\in(0,1)$, $p\in[1,\infty]$, $q\in[2,\infty)$ and $0<\eta<1$. There exists a constant $c>0$ such that
$$
m \geqslant c( s^{2-2 / q} \log (e N / s)+\log \left(\eta^{-1}\right) ).
$$

Then, with probability exceeding $1-\eta$, $\Phi$ satisfies $\omega$-RNSP$_{p,q}$ of order $s$ with parameters $\rho$ and $\tau / m^{1 / p}$ for some $0<\rho<1$ and $\tau>0$.
\end{proposition}
For the unweighted case, there is a theorem in \cite{dirk2018} that the standard Gaussian matrix will satisfy $\ell_q$-RNSP, where the $\ell_{\omega,1}$ norm is substituted for $\ell_1$ within the context of $\omega$-RNSP$_{p,q}$, with high probability for sufficiently large $m$.
With Lemma \ref{td}, the proof of Proposition \ref{nspgaussian} is available for Theorem III.3 in \cite{dirk2018}.

To analyze $\omega$-RNSP$_{p,q}$, define the cone 
$$
\begin{aligned}
& T_{\rho, s}^{q}= \\
& \left\{x \in \mathbb{C}^{n}: \exists S \subset[n],|S|=s:\left\|x_{S}\right\|_{q} \geqslant \frac{\rho}{s^{1-1 / q}}\left\|x_{S^{c}}\right\|_{\omega,1}\right\}.
\end{aligned}
$$

To establish Proposition \ref{nspgaussian}, we require the following lemmas.
\begin{lemma}[see Lemma 3.1 in \cite{dirk2018}]
\label{lemma31}
    Fix $1 \leq p<\infty$. Let $\mathcal{F}$ be a class of functions from $\mathbb{C}^n$ into $\mathbb{C}$. Consider
$$
Q_{\mathcal{F}}(u)=\inf _{f \in \mathcal{F}} \mathbb{P}(|f(X)| \geq u)
$$
and
$$
R_m(\mathcal{F})=\mathbb{E} \sup _{f \in \mathcal{F}}\left|\frac{1}{m} \sum_{i=1}^m \varepsilon_i f\left(X_i\right)\right|,
$$
where $\left(\varepsilon_i\right)_{i \geq 1}$ is a Rademacher sequence. Let $u>0$ and $t>0$, then, with probability at least $1-2 e^{-2 t^2}$,
$$
\inf _{f \in \mathcal{F}} \frac{1}{m} \sum_{i=1}^m\left|f\left(X_i\right)\right|^p \geq u^p\left(Q_{\mathcal{F}}(2 u)-\frac{4}{u} R_m(\mathcal{F})-\frac{t}{\sqrt{m}}\right).
$$
\end{lemma}

\begin{lemma}
    \label{td}
     For the given $\omega\in[\theta,1]^N$, $\theta\in(0,1)$, $1 \leqslant q<\infty$. Set

$$
\Sigma_{s}^{q}:=\left\{x \in \mathbb{C}^{n}:\|x\|_{0} \leqslant s,\|x\|_{q}=1\right\},
$$

and let $D_{s}^{q}$ be its convex hull.

Then for the unit ball $B_{\ell_{D_s^q}}$ with respect to $\|\cdot\|_{D_{s}^{q}}$, $B_{\ell_{D_s^q}}\subset D_s^q$.
As a consequence,

$$
T_{\rho, s}^{q} \cap B_{\ell_{q}^{n}} \subset\left(2+\frac{1}{\rho\theta}\right) D_{s}^{q}.
$$

\end{lemma}
Based on Lemma III.2 in \cite{dirk2018}, we make some modifications to the form of $T_{\rho, s}^{q}$.

\begin{proof}
Define 

$$
\|x\|_{D_{s}^{q}}:=\sum_{\ell=1}^{\lceil n / s\rceil}\left(\sum_{i \in I_{\ell}} x_{i}^{* q}\right)^{1 / q},
$$

where $I_{1}, \ldots, I_{\lceil n / s\rceil}$ form a uniform partition of $[n]$, i.e.,

$I_{\ell}= \begin{cases}\{s(\ell-1)+1, \ldots, s \ell\}, & \ell=1, \ldots,\lceil n / s\rceil-1, \\ \{s(\lceil n / s\rceil-1)+1, \ldots, n\}, & \ell=\lceil n / s\rceil,\end{cases}$

and $x^{*}$ is the nonincreasing rearrangement of $x$ based the value of $x\cdot \omega$.

 Suppose that $\|x\|_{D_{s}^{q}} \leqslant 1$. We make a nonincreasing rearrangement $x^*$ of the entries in the set $[N]$ based on the values of $x\cdot\omega$,  and partition $[N]$ as $S_1\cup S_2\cup\ldots$, with $|S_{i}|=s$. So that $|x_j^*|\omega_j\leqslant|x_k^*|\omega_k$ for all $j\in S_l$, $k\in S_{l-1}$, $l\geqslant2$. Let $\alpha_i=\|x_{S_i}\|_q$, $i\geqslant1$.

Then $x$ can be represented as

$$
x=\sum_{i} \alpha_{i}\left(\alpha_{i}^{-1} x_{S_{i}}\right)
$$

where

$$
\sum_{i} \alpha_{i}=\sum_{i}\left\|x_{S_{i}}\right\|_{q}=\|x\|_{D_{s}^{q}} \leqslant 1.
$$

And $\left\|\alpha_{i}^{-1} x_{S_{i}}\right\|_{q}=\frac{\left\|x_{S_{i}}\right\|_{q}}{\left\| x_{S_{i}}\right\|_{q}}=1$, $\left\|\alpha_{i}^{-1} x_{S_{i}}\right\|_{0} \leqslant s$, so $x \in D_{s}^{q}$, and $B_{\ell_{D_s^q}}\subset D_s^q$.

Assume that the vector $x\in T_{\rho, s}^{q} \cap B_{\ell_{q}^{n}}$, then it can be regarded as 
$$
    \|x\|_{D_{s}^{q}}=\left(\sum_{i \in I_{1}} x_{i}^{* q}\right)^{1 / q}+\left(\sum_{i \in I_{2}} x_{i}^{* q}\right)^{1 / q}+\sum_{\ell \geqslant 3}\left(\sum_{i \in I_{\ell}} x_{i}^{* q}\right)^{1 / q}.
$$

For $i\in I_l$, $l\geqslant3$, $|x_i^*|\omega_i\leqslant\frac{1}{s}\sum_{j\in I_{l-1}}|x_j^*|\omega_j$.

So
\begin{align*}
    \sum_{l\geqslant3}(\sum_{i\in I_l} (x^*_i)^q)^{1/q}&\leqslant\sum_{l\geqslant3}(\sum_{i\in I_l}(\frac{1}{s}\sum_{j\in I_{l-1}}|x_j^*|\omega_j)^q(\omega_i^{-1})^q)^{1/q};\\
    &\leqslant \sum_{l\geqslant3}\frac{s^{1/q-1}}{\theta}\sum_{j\in I_{l-1}}|x_j^*|\omega_j;\\
    &\leqslant \frac{s^{1/q-1}}{\theta}\sum_{l\geqslant2}\sum_{j\in I_{l}}|x_j^*|\omega_j;\\
    &\leqslant \frac{s^{1/q-1}}{\theta}\|x_{S^c}\|_{\omega,1}.
\end{align*}

The last inequality holds since $I_1$ contains the largest $s$ entries of $x\cdot\omega$, $\|x^*_{I_1^c}\|_{\omega,1}\leqslant\|x_{S^c}\|_{\omega,1}$ for any set $S \subset[n]$ with $|S|=s$.

Since $x\in T_{\rho, s}^{q} \cap B_{\ell_{q}^{n}}$, there exist a set $S \subset[n]$ with $|S|=s$, such that $\left\|x_{S}\right\|_{q} \geqslant \frac{\rho}{s^{1-1 / q}}\left\|x_{S^{c}}\right\|_{\omega,1}$.

So
$$
\sum_{l\geqslant3}(\sum_{i\in I_l} (x^*_i)^q)^{1/q}\leqslant\frac{1}{\rho\theta}\left\|x_{S}\right\|_{q}\leqslant\frac{1}{\rho\theta}(\sum_{i\in I_1} (x^*_i)^q)^{1/q}.
$$

Since $\|x\|_q\leqslant1$, $\|x\|_{D_{s}^{q}} \leqslant 2+\frac{1}{\rho\theta}$.
\end{proof}

\begin{proof}[Proof of Proposition \ref{nspgaussian}]
  To prove the proposition, a sufficient condition of  $\omega$-RNSP$_{p,q}$ is introduced. We need to exhibit that if
\begin{equation}
\label{pb}
\mathbb{P}\left(\inf _{x \in T_{\rho, s}^{q} \cap S_{\ell_{q}^{n}}}\|\Phi x\|_{p} \geqslant \frac{m^{1 / p}}{\tau}\right) \geqslant 1-\eta 
\end{equation}
holds, where $S_{\ell_{q}^{n}}$ is the sphere of $\ell_q$ norm, then with probability exceeding $1-\eta$, $\Phi$ satisfies $\omega$-RNSP$_{p,q}$.\\

Suppose \eqref{pb} holds, for any $x \in \mathbb{C}^{n}$, $S \subset[N]$ with $|S| \leqslant s$, there are two cases.

If $\|A x\|_{p}<\left(m^{1 / p} / \tau\right)\|x\|_{q}$, then $x /\|x\|_{q}$ is not in $T_{\rho, s}^{q}$,

$$
\left\|x_{S}\right\|_{q} \leqslant \frac{\rho}{s^{1-1 / q}}\left\|x_{S^{c}}\right\|_{\omega,1} \leqslant \frac{\rho}{s^{1-1 / q}}\left\|x_{S^{c}}\right\|_{\omega,1}+\frac{\tau}{m^{1 / p}}\|A x\|_{p} .
$$

If $\|A x\|_{p} \geqslant$ $\left(m^{1 / p} / \tau\right)\|x\|_{q}$, then 

$$
\left\|x_{S}\right\|_{q} \leqslant\|x\|_{q} \leqslant \frac{\rho}{s^{1-1 / q}}\left\|x_{S^{c}}\right\|_{\omega,1}+\frac{\tau}{m^{1 / p}}\|A x\|_{p} .
$$

To prove \eqref{pb}, we write

$$
\inf _{x \in T_{\rho, s}^{q} \cap S_{\ell_{q}^{n}}} \frac{\|A x\|_{p}}{m^{1 / p}}=\inf _{x \in T_{\rho, s}^{q} \cap S_{q}^{n}}\left(\frac{1}{m} \sum_{i=1}^{m}\left|\left\langle X_{i}, x\right\rangle\right|^{p}\right)^{1 / p},
$$

where $X_{i}$ denotes the $i$-th row of $A$. To apply Lemma \ref{lemma31}, we estimate the small ball probability $Q_{\mathcal{F}}$ and the expected Rademacher supremum $R_{m}(\mathcal{F})$ for the set of linear functions

$$
\mathcal{F}=\left\{\langle\cdot, x\rangle: x \in T_{\rho, s}^{q} \cap S_{\ell_{q}^{n}}\right\}
$$

Let $V=m^{-1 / 2} \sum_{i=1}^{m} \varepsilon_{i} X_{i}$, then by Lemma \ref{td},

$$
\begin{aligned}
R_{m}(\mathcal{F}) & =m^{-1 / 2} \mathbb{E} \sup _{x \in T_{\rho, s}^{q} \cap S_{\ell_{q}^{n}}}\langle V, x\rangle \\
& \leq\left(2+\frac{1}{\rho\theta}\right) m^{-1 / 2} \mathbb{E} \sup _{x \in D_{s}^{q}}\langle V, x\rangle \\
& =\left(2+\frac{1}{\rho\theta}\right) m^{-1 / 2} \mathbb{E} \sup _{x \in \Sigma_{s}^{q}}\langle V, x\rangle,
\end{aligned}
$$

as $D_{s}^{q}$ is the convex hull of $\Sigma_{s}^{q}$. Since any $x \in \Sigma_{s}^{q}$ satisfies $\|x\|_{2} \leq s^{1 / 2-1 / q}\|x\|_{q}=s^{1 / 2-1 / q}$,

$$
R_{m}(\mathcal{F}) \leq s^{1 / 2-1 / q}\left(2+\frac{1}{\rho\theta}\right) m^{-1 / 2} \mathbb{E} \sup _{x \in \Sigma_{s}^{2}}\langle V, x\rangle
$$

Since $X_{1}, \ldots, X_{m}$ are independent standard Gaussian vectors, so is $V$. With Lemma 4 in \cite{kaba2015}, $
\mathbb{E} \sup _{x \in \Sigma_{s}^{2}}\langle V, x\rangle
$, the Gaussian width of $\Sigma_{s}^{2}$, is known to be

$$
\mathbb{E} \sup _{x \in \Sigma_{s}^{2}}\langle V, x\rangle \leq \sqrt{2 s \log (e n / s)}+\sqrt{s}
$$

and we can conclude that

$$
R_{m}(\mathcal{F}) \leq c s^{1-1 / q}\left(2+\frac{1}{\rho\theta}\right) m^{-1 / 2} \sqrt{\log (e n / s)} .
$$

To estimate the small ball probability, note that, since $\|x\|_{q} \leq\|x\|_{2}$, for any $x \in S_{\ell_{q}^{n}}$,

$$
\begin{aligned}
\mathbb{P}\left(\left|\left\langle X_{i}, x\right\rangle\right| \geq u\right) & =\mathbb{P}\left(\left|\left\langle X_{i}, \frac{x}{\|x\|_{2}}\right\rangle\right| \geq \frac{u}{\|x\|_{2}}\right) \\
& \geq \mathbb{P}\left(\left|\left\langle X_{i}, \frac{x}{\|x\|_{2}}\right\rangle\right| \geq u\right)=\mathbb{P}(|g| \geq u),
\end{aligned}
$$

where $g$ is a standard Gaussian real-valued random variable. Therefore,

$$
Q_{\mathcal{F}}(2 u) \geq \mathbb{P}(|g| \geq 2 u)
$$

Now pick $u_{*}$ small enough so that the right hand side is bigger than $1 / 2$, say. Pick $m$ large enough so that

$$
\max \left\{\frac{4 c\left(2+\frac{1}{\rho\theta}\right) s^{1-1 / q} \sqrt{\log (e n / s)}}{u_{*} \sqrt{m}}, \frac{\sqrt{\log (2 / \eta)}}{\sqrt{2 m}}\right\} \leq 1 / 8
$$

By Lemma \ref{lemma31}, we can now conclude that \eqref{pb} holds with $\tau=4^{1 / p} / u_{*}$.

Finally, let $p=\infty$. Since $\|A x\|_{\log m} \leq e\|A x\|_{\infty}$,

$$
\begin{aligned}
& \mathbb{P}\left(\inf _{x \in T_{\rho, s}^{q} \cap S_{\ell_{q}^{n}}}\|A x\|_{\infty} \geq \frac{1}{\tau}\right) \\
& \geq \mathbb{P}\left(\inf _{x \in T_{\rho, s}^{q} \cap S_{\ell_{q}^{n}}}\|A x\|_{\log m} \geq \frac{e}{\tau}\right)
\end{aligned}
$$

Thus, in this case the result follows from our proof for $p=$ $\log m$.
\end{proof}

In both of the aforementioned propositions, the specified bound for the number of observations $m$ for RIP$_{p,2}$ is larger than that in Proposition 
\ref{nspgaussian}. This indicates a gap between the two properties. Therefore, the following section will examine the relationship between the two properties.

\begin{theorem}
    \label{rip2nsp}
    For any given weight $\omega=[\omega_1,\ldots,\omega_N]^T\in [0,1]^N$, if the matrix $\Phi$ satisfie$RIP_{p,q}$ of order $2s$ with $\delta_{2s}$,$s\geqslant2$ and $p,q\in[2,+\infty)$, then $\Phi$ satisfies $\omega$-RNSP$_{p,q}$ of order $s$ with constants  $\rho=\frac{C_{p,q}}{(1-\delta_{2s})^{2/q}}$ and $\gamma=\frac{(1+\delta_{2s})^{1/q}}{\mu_{p,q}(1-\delta_{2s})^{2/q}}$, where $C_{p,q}$ is defined in  Lemma \ref{lemma2}.
\end{theorem}
\begin{proof}[Proof of Theorem \ref{rip2nsp}]
For any vector $\nu\in\mathbb{C}^N$ and index set $S =: S_0$ with $|S|= s$. 
We make a nonincreasing rearrangement of the entries in set $S^{c}$ based on the values of $\nu_{i}\omega_{i}^{-1}$ with $i\in S^c$, and partition $S_0^c$ as $\{S_l:1\leqslant l\leqslant\lceil\frac{N-s}{s}\rceil \}$ with $|S_l|=s$ for $1\leqslant l\leqslant\lfloor\frac{N-s}{s}\rfloor$ and $|S_l|\leqslant s$ for $\lceil\frac{N-s}{s}\rceil$. So that $|\nu_j|\omega_j^{-1}\leqslant|\nu_k|\omega_k^{-1}$ for all $j\in S_l$, $k\in S_{l-1}$, $l\geqslant2$. We denote that $S_{01}=S_0\cup S_1$.

   According to the definition of the matrix $J$ in Lemma \ref{lemma2},
   $$
    \|\Phi \nu_{S_{01}}\|_p^2=\langle J(\Phi \nu_{S_{01}}),\Phi \nu\rangle-\langle J(\Phi \nu_{S_{01}}),\Phi\left(\sum_{l\geqslant2}\nu_{S_l}\right)\rangle.
    $$
   
     By the H$\ddot{o}$lder inequality and $\omega$-$RIP_{p,q}$ on $S_{01}$,
 \begin{align*}
    |\langle J(\Phi \nu_{S_{01}}),\Phi \nu\rangle|&\leqslant\|J(\Phi \nu_{S_{01}})\|_{\frac{p}{p-1}}\|\Phi \nu\|_p\\
    &=\|\Phi \nu_{S_{01}}\|_p\|\Phi \nu\|_p\\
    &\leqslant \|\Phi \nu\|_p\mu_{p,q}(1+\delta_{2s})^{1/q}\|\nu_{S_{01}}\|_q.
  \end{align*}
 When $l\geqslant2$, $S_{01}$ and $S_l$ are disjoint. By Lemma \ref{lemma2}, if the matrix $\Phi$ satisfies $\omega$-$RIP_{p,q}$ of order $2s$, then
         $$
|\langle J(\Phi \nu_{S_{01}}),\Phi \nu_{S_l}\rangle|\leqslant \mu_{p,q}^2C_{p,q}\|\nu_{S_{01}}\|_q\|\nu_{T_l}\|_q.
    $$
    Combining the equations above, we have
     \begin{align*}
    \mu_{p,q}^2(1-\delta_{2s})^{2/q}\|\nu_{S_{01}}\|_q^{2}&\leqslant\|\Phi \nu_{S_{01}}\|_p^2\\
    &\leqslant \|\Phi \nu\|_p\mu_{p,q}(1+\delta_{2s})^{1/q}\|\nu_{S_{01}}\|_q\\
    &\quad+\mu_{p,2}^2C_{p,q}\|\nu_{S_{01}}\|_q\sum_{j\geqslant2}\|\nu_{S_l}\|_q
  \end{align*}
    Divide both sides by$\|\nu_{S_{01}}\|_q$,
    \begin{equation}
    \label{l2tl2}
        \|\nu_{S_{01}}\|_q\leqslant\frac{(1+\delta_{2s})^{1/q}}{\mu_{p,q}(1-\delta_{2s})^{2/q}}\|\Phi \nu\|_p+\frac{C_{p,q}}{(1-\delta_{2s})^{2/q}}\sum_{l\geqslant2}\|\nu_{S_l}\|_q.
    \end{equation}
     
     For arbitrary $k\in S_l$, assume the parameters $\alpha_k=(\sum_{j\in T_l}\omega_j^2)^{-1}\omega_k^2\leqslant(s-1)^{-1}\omega_k^2$, obviously that $\sum_{k\in T_l}\alpha_{k}=1$.\\
$$
      \forall j\in S_l,\,l\geqslant 2,\,|\nu_j|\omega_j^{-1}\leqslant \sum_{k\in S_{l-1}}\alpha_k|\nu_k|\omega_k^{-1}\leqslant (s-1)^{-1}\sum_{k\in S_{l-1}}|\nu_k|\omega_k.
    $$
    By Cauchy-Schwartz inequality and $q\geqslant2$, for any $l\geqslant 2$, 
    \begin{align*}
        \|\nu_{S_l}\|_q&=(\sum_{j\in S_l}(|\nu_j|\omega_j^{-1})^q\omega_j^q)^{\frac{1}{q}}\\
       &\leqslant\frac{s^{1/q}}{s\theta^2}\|\nu_{S_{l-1}}\|_{\omega,1}\\
        &\leqslant \frac{1}{\theta^2}s^{1/q-1}\|\nu_{S_{l-1}}\|_{\omega,1},
    \end{align*}
   
\begin{align*}
    \|\nu_{S_{01}^c}\|_q&\leqslant \sum_{l\geqslant 2}\|\nu_{S_{l}}\|_q\leqslant\sum_{l\geqslant 2}\|\nu_{S_{l}}\|_2\\
    &\leqslant\frac{1}{\theta^2}s^{1/q-1}(\|\nu_{S_1}\|_{\omega,1}+\|\nu_{S_2}\|_{\omega,1}+\ldots)\\
    &\leqslant\frac{1}{\theta^2}s^{1/q-1}\|\nu_{S^c}\|_{\omega,1} .
\end{align*}

  Substitute $\|\nu_{S_l}\|_q$ in \eqref{l2tl2} with the above estimate, we have  
    \begin{align*}
    \label{hs2}
        \|\nu_S\|_q&\leqslant\|\nu_{S_{01}}\|_q\\
       &\leqslant\frac{(1+\delta_{2s})^{1/q}}{\mu_{p,q}(1-\delta_{2s})^{2/q}}\|\Phi \nu\|_p+\frac{C_{p,q}}{\theta^2(1-\delta_{2s})^{2/q}}s^{1/q-1}\|\nu_{S^c}\|_{\omega,1}.
    \end{align*}

    So $\rho=\frac{C_{p,q}}{\theta^2(1-\delta_{2s})^{2/q}}$ and $\gamma=\frac{(1+\delta_{2s})^{1/q}}{\mu_{p,q}(1-\delta_{2s})^{2/q}}$.
\end{proof}

\section{Numerical Implementation}
\label{sec:exp}
In this section, proximal operators are introduced to provide specific algorithms for $\omega$-BPDQ$_p$. Additionally, experiments for test signals will be presented.
\subsection{Proximal Optimization}

 The  $\omega$-BPDQ$_p$ model can be regarded as a special case of the convex problem formulated as:
\begin{equation}
\label{convexpro}
    \arg \min _{x \in \mathcal{H}} f_{1}(x)+f_{2}(x)
\end{equation}
where $\mathcal{H}=\mathbb{R}^{N}$ is seen as a Hilbert space equipped with the inner product $\langle x, z\rangle=\sum_{i} x_{i} z_{i}$.

 In $\omega$-BPDQ$_p$, $f_1(x)=\|x\|_{\omega,1}$ and $f_{2}(x)=\imath_{T^{p}(\epsilon)}(x)=\begin{cases}
  0 \text{ if } x\in {T^{p}(\epsilon)}\\
  \infty \text{ if } x\notin {T^{p}(\epsilon)}
\end{cases}$, where ${T^{p}(\epsilon)}=\left\{x \in \mathbb{R}^{N}:\left\|y_{\mathrm{q}}-\Phi x\right\|_{p} \leqslant \epsilon\right\}$.

The domain of a function $f:\mathcal{H}\to(-\infty,+\infty]$ is defined as $\operatorname{dom} f=\{x \in \mathcal{H}: f(x)<+\infty\}$. If a convex function $f:\mathcal{H}\to(-\infty,+\infty]$ is lower semi-continuous, i.e., $\liminf _{x \rightarrow x_{0}} f(x)=f\left(x_{0}\right)$ for all $x_{0} \in \operatorname{dom} f$, and $\operatorname{dom} f\ne \emptyset$, which means $f$ is not identically equal to $\infty$ everywhere, then we denote the function space consisting of all lower semi-continuous functions $f$ as $\Gamma_{0}(\mathcal{H})$.

To solve \eqref{convexpro}, we will introduce the proximal mapping, defined as the unique solution $\operatorname{prox}_{f}(x)=\arg \min _{z \in \mathcal{H}} \frac{1}{2}\|z-x\|_{2}^{2}+$ $f(z)$ associated with a closed convex function $f\in\Gamma_{0}(\mathcal{H})$\cite{jac2011}\cite{j2014A}\cite{drsplliting}. We recall that the subdifferential of a function $f \in \Gamma_{0}(\mathcal{H})$ at $x \in \mathcal{H}$ is $\partial f(x)=\{u \in \mathcal{H}$ : $\forall z \in \mathcal{H}, f(z) \geqslant f(x)+\langle u, z-x\rangle\}$.  

\begin{proposition}[see Theorem B.1 in \cite{j2014A}]
\label{b1}
     A vector $x$ is a minimum of a convex function $f$ if and only if $0\in\partial f(x)$.
\end{proposition}
\begin{proposition}[see Theorem B.3 in \cite{j2014A}]
\label{b3}
    Let $f\in\Gamma_{0}(\mathcal{H})$, then $y=\operatorname{prox}_{f}(x)$ if and only if $x\in y+\partial f(y)$.
\end{proposition}

With Proposition \ref{b3}, $0\in\partial (f_1+f_2)(x)$, $x$ is also the solution of  \eqref{convexpro}. And it still holds in reverse.

In practice, it is easier to evaluate the proximal mappings of $f_1$ and $f_2$ separately than directly compute $\operatorname{prox}_{\beta (f_1+f_2)}(x)$. Since $f_1$ and $f_2$ are nondifferentiable in  $\omega$-BPDQ$_p$, and the proximal mapping is not linear, a particular monotone operator splitting method known as the Douglas–Rachford (DR) splitting method is applicable in this case\cite{drsplliting}.

\begin{equation}
\label{dr}
    x_{k+1}=\left(1-\frac{\alpha_{k}}{2}\right) x_k+\frac{\alpha_{k}}{2} S_{\gamma}^{\odot} \circ \mathcal{P}_{T_{p}(\epsilon)}^{\odot}\left(x_k\right),
\end{equation}

where $A^{\odot} := 2 A-Id$ for any operator $A, \alpha_{k} \in(0,2)$ for all $k \in \mathbb{N}, S_{\gamma}=\operatorname{prox}_{\gamma f_{1}}$ with $\gamma>0$ and $\mathcal{P}_{T_{p}(\epsilon)}=\operatorname{prox}_{f_{2}}$.

The next objective is to show that the DR splitting method will converge to the solution of \eqref{convexpro}.
Simplified \eqref{dr}, we have 
$$
x_{k+1}=x_k+\alpha_k\operatorname{prox}_{\gamma f_1}(2\operatorname{prox}_{f_2}(x_k)-x_k)-\alpha_k\operatorname{prox}_{f_2}(x_k).
$$

Denote $F(x)=x+\operatorname{prox}_{g}(2\operatorname{prox}_{f}(x)-x)-\operatorname{prox}_{f}(x)$
and $G(x)=x-F(x)=\operatorname{prox}_{f}(x)-\operatorname{prox}_g(2\operatorname{prox}_{f}(x)-x)$, we have the fixed point iteration with relaxation $\alpha$
\begin{equation}
    \begin{split}
\label{fpi}
    x_{k+1}&=(1-\alpha_k)x_k+\alpha_kF(x_k)\\
    &=x_k-\alpha_kG(x_k).
      \end{split}
\end{equation}
with $g=\gamma f_1$ and $f=f_2$.

It is necessary to prove that the fixed point iteration converges to a point, and the solution of \eqref{convexpro} is associated with the fixed point of $F(x)$. A similar proof can be found in \cite{Combettes2004}. For the sake of completeness, another proof is provided below.
\begin{proposition}
    \label{fpconv}
    If in the fixed point iteration \eqref{fpi}, $F$ has fixed points and $\alpha_k\in[\alpha_{\min},\alpha_{\max}]$ with $0<\alpha_{\min}<\alpha_{\max}<2$, then $x_k$ converges to a fixed point $x^*$ of \eqref{fpi} and $y_{k+1}=\operatorname{prox}_f(x_k)$ converges to $y^*=\operatorname{prox}_f(x^*)$.
\end{proposition}
\begin{proof}
  Define $y=\operatorname{prox}_f(x)$ and $\nu=\operatorname{prox}_g(2y-x)$, $F$ can be expressed as $F(x)=x+\nu-y$.

   Since proximal mappings are firmly nonexpansive\cite{rauhut2016}, for any $x$, $\hat{x}$, $y$ and $\hat{y}$, we have
   \begin{align*}
   \label{fnexpansive}
       (y-\hat{y})^T(x-\hat{x})&\geqslant\|y-\hat{y}\|_2^2.\\
       (\nu-\hat{\nu})^T(2y-x-2\hat{y}+\hat{x})&\geqslant\|\nu-\hat{\nu}\|_2^2.
   \end{align*}

   Then 
\begin{align*}
    (F(x)-F(\hat{x}))^T(x-\hat{x})&=(x+\nu-y-\hat{x}-\hat{\nu}+\hat{y})^T(x-\hat{x})\\
    &\geqslant(x+\nu-y-\hat{x}-\hat{\nu}+\hat{y})^T(x-\hat{x})\\
    &\quad-(y-\hat{y})^T(x-\hat{x})+\|y-\hat{y}\|_2^2\\
    &=(\nu-\hat{\nu})^T(x-\hat{x})+\|x-y+\hat{x}-\hat{y}\|_2^2\\
    &=(\nu-\hat{\nu})^T(2y-x-2\hat{y}+\hat{x})\\
    &\quad-\|\nu-\hat{\nu}\|_2^2+\|F(x)-F(\hat{x})\|_2^2\\
    &\geqslant\|F(x)-F(\hat{x})\|_2^2.
\end{align*}

Similarly, 
\begin{align*}
    \|G(x)-G(\hat{x})\|_2^2&=(G(x)-G(\hat{x}))^T(G(x)-G(\hat{x}))\\
    &=(x-\hat{x}-F(x)+F(\hat{x}))^T(x-\hat{x}-F(x)+F(\hat{x}))\\
    &=\|x-\hat{x}\|_2^2+\|F(x)-F(\hat{x})\|_2^2-(x-\hat{x})^T(F(x)\\
    &\quad-F(\hat{x}))-(F(x)-F(\hat{x}))^T(x-\hat{x})\\
    &\leqslant\|x-\hat{x}\|_2^2-\|F(x)-F(\hat{x})\|_2^2\\
    &\leqslant(G(x)-G(\hat{x}))^T(x-\hat{x}).
\end{align*}

Suppose that $x^*$ is a fixed point of $F$, i.e. $F(x^*)=x^*$, and $G(x^*)=0$.
Consider the $k_{th}$ iteration:
\begin{align*}
    &\quad\|x_{k+1}-x^*\|_2^2- \|x_{k}-x^*\|_2^2\\
    &=2(x_{k+1}-x_k)^T(x_k-x^*)+\|x_{k+1}-x_k\|_2^2\\
    &\leqslant-2\alpha_kG(x_k)^T(x_k-x^*)+\alpha_k^2\|G(x_k)\|_2^2\\
    &\leqslant-\alpha_k(2-\alpha_k)\|G(x_k)\|_2^2\\
    &\leqslant-M\|G(x_k)\|_2^2.
\end{align*}
where $M=\alpha_{\min}(2-\alpha_{\max})$.

$M\|G(x_k)\|_2^2\leqslant \|x_{k}-x^*\|_2^2- \|x_{k+1}-x^*\|_2^2$ implies that $M\sum_{k=0}^{\infty}\|G(x_k)\|_2^2\leqslant\|x_0-x^*\|_2^2$. The series $\sum_{k=0}^{\infty}\|G(x_k)\|_2^2$ is bounded, so $\|G(x_k)\|_2^2$ converges to $0$.

And since $\|x_{k+1}-x^*\|_2^2- \|x_{k}-x^*\|_2^2\leqslant0$, $\|x_{k}-x^*\|_2^2$ is nonincreasing, then the sequence $\{x_k\}_{k=0}^{\infty}$ is bounded. There exists a subsequence $\{\bar{x}_k\}_{k=0}^{\infty}$ that converges to $\bar{x}$.

By the continuity of $G$, $0=\lim_{k\to\infty}G(\bar{x}_k)=G(\bar{x})$. Then $\bar{x}$ is a zero of $G$ and $\lim_{k\to\infty}\|x_k-\bar x\|_2$ exists. So $\bar{x}$ is a fixed point of $F$.

Then the uniqueness of the solution will be exhibited. Suppose that $\bar{u}$ and $\bar{v}$ are two different limit points, and there exist subsequences of $\{\bar{x}_k\}_{k=0}^{\infty}$ converge to $\bar{u}$ and $\bar{v}$ respectively. Since $\lim_{k\to\infty}\|x_k-\bar u\|_2$ and $\lim_{k\to\infty}\|x_k-\bar v\|_2$ exist, 
$$
    \|\bar{u}-\bar{v}\|_2=\lim_{k\to\infty}\|x_k-\bar u\|_2=\lim_{k\to\infty}\|x_k-\bar v\|_2=0.
$$

There is a contradiction with $\bar{u}\neq\bar{v}$, so $\bar{x}=x^*$ is the unique fixed point of $F$.
\end{proof}

Since $F(x)$ is continuous, $x^*$ is the fixed point of $F(x)$, i.e., $x^*=F(x^*)$. It can be shown that $\operatorname{prox}_{f_2}(x^*)$ is the solution of \eqref{convexpro} in the following proposition.

\begin{proposition}
    \label{proxx*}
    $x^*$ is a fixed point of $F$ if and only if $y=\operatorname{prox}_{f_2}(x^*)$ satisfies $0\in\partial f_1(y)+\partial f_2(y)$.
\end{proposition}
\begin{proof}
     $x^*$ is a fixed point of $F$ can be represented as 
    \begin{align*}
        F(x^*)&=x^*\\
        x^*+\operatorname{prox}_{\gamma f_1}(2\operatorname{prox}_{f_2}(x^*)-x^*)-\operatorname{prox}_{f_2}(x^*)&=x^*\\
        y&=\operatorname{prox}_{\gamma f_1}(2y-x^*)
    \end{align*}

    According to Proposition \ref{b3}, $ y=\operatorname{prox}_{\gamma f_1}(2y-x^*)$ if and only if 
    \begin{equation}
        \label{e1}
        y-x^*\in\partial f_1(y).
    \end{equation}

    And $y=\operatorname{prox}_{f_2}(x^*)$ if and only if  \begin{equation}
        \label{e2}
        x^*-y\in\partial f_2(y).
    \end{equation}
    Add \eqref{e1} and \eqref{e2}, we have $0\in f_1(y)+f_2(y)$.

Therefore if we have the fixed point that $F(x^*)=x^*$, then $y=\operatorname{prox}_{f_2}(x^*)$ satisfies $0=(y-x^*)+(x^*-y)\in f_1(y)+f_2(y)$. 
Conversely, if there exists $z$ such that $z\in\partial f_1(y)$ and $-z\in\partial f_2(y)$, then $x^*=y-z$ is a fixed point of $F$.
\end{proof}

In the subsequent part of this section, we will present the proximal operators of specific functions.

In $\omega$-BPDQ$_p$, $f_1(x)=\|x\|_{\omega,1}$, which can be interpreted as $\|Wx\|_1$, where $W$ is a diagonal matrix with the entries of the weight $\omega$ on its diagonal.

Let $R(x)=\|x\|_1$, since Proposition \ref{b3}, $u=\operatorname{prox}_R(x)$ if and only if $x-u\in\partial R(u)$. It is easy to compute $\partial R(u)$ with its components as follows:
$$
\partial |u_i|=\begin{cases}
  1& \text{ if } u_i>0 \\
  -1& \text{ if } u_i<0 \\
  [-1,1]& \text{ if } u_i=0.
\end{cases}
$$

So for any constant $\gamma>0$, $\operatorname{prox}_{\gamma R}(x)$ can be obtained by the soft thresholding operator:
$$
(\operatorname{prox}_{\gamma R}(x))_i=\begin{cases}
  x_i-\gamma& \text{ if } ux_i>\gamma \\
  0& \text{ if } |x_i|\leqslant\gamma \\
  x_i+\gamma& \text{ if } x_i<-\gamma.
\end{cases}
$$

To compute the proximal operator of the weighted norm, it is necessary to find the proximal operator of the function composition.

\begin{proposition}[see Proposition 11 in \cite{drsplliting}]
   Let $f \in \Gamma_0(\mathcal{H})$ and $L$ be a bounded linear operator that satisfies $L \circ L^*=\nu \mathrm{Id}$ for some $\nu \in(0,+\infty)$.

Then the proximal operator of $f \circ L$ can be represented as
$$
\operatorname{prox}_{f \circ L}=\mathrm{Id}+\nu^{-1} L^* \circ\left(\operatorname{prox}_{\nu f}-\mathrm{Id}\right)\circ L .
$$

\end{proposition}
\begin{proof}

Set $g=f \circ L$ and $p=\operatorname{prox}_g (x)$ for any $x \in \mathcal{H}$, then $\partial g=L^*\circ\partial f \circ L$. 
From Proposition \ref{b3}, $p=\operatorname{prox}_g(x)$ if and only if $x-p \in \partial g(p)=L^*(\partial f(L p))$.
Then it implies that 
\begin{align*}
    L(x-p) &\in L\left(L^*(\partial f(L p))\right)\\
    L x-L p &\in \nu \partial f(L p)=\partial(\nu f)(L p) \\
    L p&=\operatorname{prox}_{\nu f}(L x) .
\end{align*}

Let $V=\operatorname{ker} L$, then $V^{\perp}=$ $\operatorname{ran} L^*$, and the projection of them are
$$
\begin{array}{l}
P_V=\mathrm{Id}-\left(L^* \circ\left(L \circ L^*\right)^{-1} \circ L\right)=\mathrm{Id}-\nu^{-1} L^* \circ L \\
P_{V^{\perp}}=\mathrm{Id}-P_V=\nu^{-1} L^* \circ L .
\end{array}
$$

 Since that $p-x\in L^*(\partial f(L p)) \in \operatorname{ran} L^*$, we have
\begin{align*}
    P_V p&=P_V x+P_V(p-x)=P_V x=\left(\operatorname{Id}-\nu^{-1} L^* \circ L\right) x.\\
    P_{V^{\perp}} p&=\nu^{-1} L^*(L p)=\nu^{-1}\left(L^* \circ \operatorname{prox}_{\nu f} \circ L\right) x.\\
    p&=P_V p+P_{V^{\perp}} p=x+\nu^{-1}\left(L^* \circ\left(\operatorname{prox}_{\nu f}-\mathrm{Id}\right) \circ L\right) x.
\end{align*}
\end{proof}

While $f_2$ can be write as 
$$f_{2}(x)=\imath_{T^{p}(\epsilon)}(x)=\begin{cases}
  0 \text{ if } x\in {T^{p}(\epsilon)}\\
  \infty \text{ if } x\notin {T^{p}(\epsilon)}
\end{cases},$$
 where ${T^{p}(\epsilon)}=\left\{x \in \mathbb{R}^{N}:\left\|y_{\mathrm{q}}-\Phi x\right\|_{p} \leqslant \epsilon\right\}$.
And ${T^{p}(\epsilon)}$ is the function composition of the unit $\ell_p$ ball $B^p=\{y\in\mathbb{R}^m:\|y\|_p\leqslant1\}\subset\mathbb{R}^m$ and the affine operator $A(x):=\frac{1}{\epsilon}(\Phi x-y_q)$, i.e. $f_2(x)=(\imath_{B^p}\circ A)(x)$. So it is necessary to compute $\operatorname{prox}_{\imath_{B^p}}$.

It is obvious that if $f(x)$ is the indicator function of closed convex set $\mathcal{C}$, then $\operatorname{prox}_f(x)$ is the projection on $\mathcal{C}$, i.e. 
$$
    \operatorname{prox}_f(x)=\arg\min_{u\in\mathcal{C}}\|u-x\|_2^2=P_{\mathcal{C}}(x).
$$

If $p=2$, the projection on the unit $\ell_2$ ball is
$$P_{B^2}(x)=
\begin{cases}
    x\text{ if } \|x\|_2<1,\\
    \frac{x}{\|x\|_2}\text{ if }\|x\|_2\geqslant1.
\end{cases}
$$
Similarly, if $2<p<\infty$, $P_{B^p}(x)=x$ when $\|x\|_p<1$.  If $\|x\|_p\geqslant1$, the projection can be expressed by 
$P_{B^p}(x)=\arg\min_{\|u\|_p=1}\|u-x\|_p^p.$ So it has the same solution with the Lagrange function as follows: for $\lambda\in\mathbb{R}$,
\begin{equation}
    \label{lage}
     L(x,u,\lambda)=\frac{1}{2}\|u-x\|_2^2+\lambda(\|u\|_p^p-1).
\end{equation}

Denote the solution of \eqref{lage} as $u^*=P_{B^p}(x)$ and $\lambda^*$. And they are required to satisfy that 
\begin{equation}
\begin{split}
     \label{pareq0}
    \frac{\partial}{\partial u}L(x,u^*,\lambda^*)&=\frac{\partial}{\partial u}\frac{1}{2}\|u^*-x\|_2^2+\frac{\partial}{\partial u}\lambda(\|u^*\|_p^p-1)=0, \\
    \frac{\partial}{\partial\lambda}L(x,u^*,\lambda^*)&=\|u^*\|_p^p-1=0.
\end{split}
\end{equation}

Define $z\in\mathbb{R}^{m+1}$ as $z_i=u_i$, $i=1,\ldots,m$ and $z_{m+1}=\lambda$. And $F=\nabla_zL$ such that \eqref{pareq0} if and only if $F(z^*)=0$. So 
$$F_i(z)=
\begin{cases}
   z_i-x_i+pz_{m+1}z_i^{p-1} \text{ if } i\leqslant m,\\
   (\sum_{j=1}^mz_j^p)-1 \text{ if } i=m+1.
\end{cases}
$$

Then the Newton method is introduced to find $z^*$. Given an initialization point $z^0=(\frac{x}{\|x\|_p},z^0_{m+1}\arg\min_\lambda\|F(u^0,\lambda)\|_2)$, the iterates are represented by 
\begin{equation*}
    z^{n+1}=z^n-J(z^n)^{-1}F(z^n),
\end{equation*}
where $J$ is the Jacobian matrix of $F$, where $J_{ij}=\frac{\partial F_i}{\partial z_j}$. After iterations, $z^n$ will close enough to $(u^*,\lambda^*)$.
\subsection{Experiments}

Before the start of the experiments, we will provide the experimental settings. Let the text signal $x\in\mathbb{R}^N$ be $k$ sparse, where $N=1024$ and $k=16$ are set in this experiment. The entries of $x$ on its support follow the standard normal distribution, and the entries out of its support are zero, i.e.
$$x
\begin{cases}
    \sim\mathcal{N}(0,1) \text{ if }i\in supp(x):=T_0\\
    =0\text{ if }i\notin supp(x).
\end{cases}
$$

The signals  are transferred via the quantizer described in \eqref{quantized}  and gained the measurement vectors $y\in\mathbb{R}^m$, where $m=256$. In the quantizer, the sensing matrix $\Phi\in\mathbb{R}^{m\times N}$ is the standard Gaussian matrix, and the bin width $\alpha=\frac{\|\Phi x\|_{\infty}}{40}$. The maximum number of DR iterations \eqref{dr} is set to 800 to reconstruct the signal. 
The quantity of reconstruction is measured by SNR $=20\log_{10}\frac{\|x\|_2}{\|x-\hat{x}\|_2}$, where $x$ is the real signal and $\hat{x}$ is the recover vector. The quantity of reconstruction increases as the SNR increases.

The prior information support is denoted by $\widetilde{T}$, i.e. for the weight $\omega\in[0,1]^N$, 
$$
\omega_i=
\begin{cases}
    \gamma \text{ if }i\in\widetilde{T}\\
    1 \text{ if }i\notin\widetilde{T}
\end{cases},
$$
where constant $\gamma\in(0,1)$. While $x$ is $k$ sparse, $|T_0|=k$, and we suppose that $|\widetilde{T}|=\rho k$ with $\rho>0$. We separate $\widetilde{T}$ into $T_1=T_0\cap \widetilde{T}$ and $T_2=\{[N]\setminus T_0\}\cap\widetilde{T}$ so that $T_1$ and $T_2$ are uniformly chosen from $T_0$ and $[N]\setminus T_0$ respectively. Let $|T_1|=\alpha|T_0|=\alpha\rho k$, where $\alpha\rho\leqslant1$. Without loss of generality, we set $\rho=1$, $\alpha=\frac{1}{2}$, and $\gamma=0.5$.

For the exact reconstruction, the measurements without noise will be reconstructed via the weighted Basis Pursuit(BP) model:
\begin{equation*}
     \min_{z\in\mathbb{C}^n}\|z\|_{\omega,1}\mbox{ subject to } \Phi z=y.
\end{equation*}

\begin{figure}

   \centering
    \begin{minipage}{.45\textwidth}
        \centering
        \includegraphics[width=\linewidth]{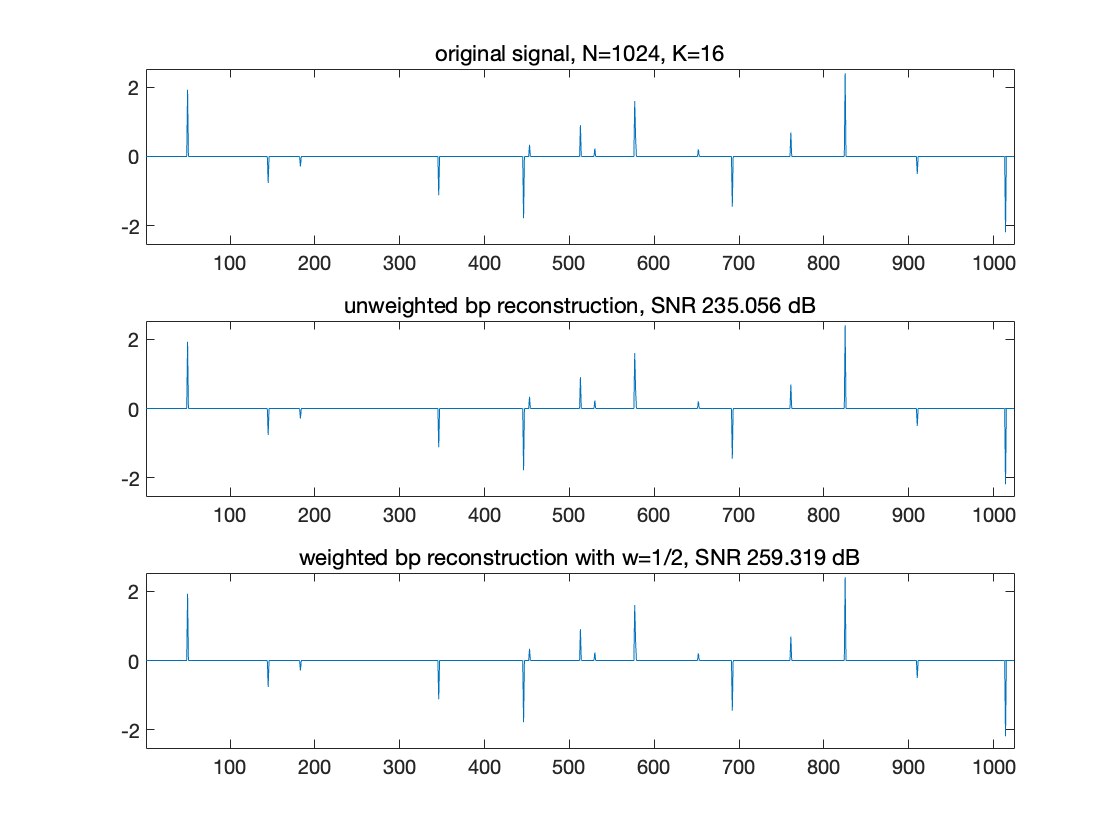}
        \caption{the original signal, the recover signals and their SNRs for Basis Pursuit and weighted Basis Pursuit with $\gamma=0.5$}
        \label{wbp1}
    \end{minipage}
    \hfill 
    \begin{minipage}{.45\textwidth}
        \centering
        \includegraphics[width=\linewidth]{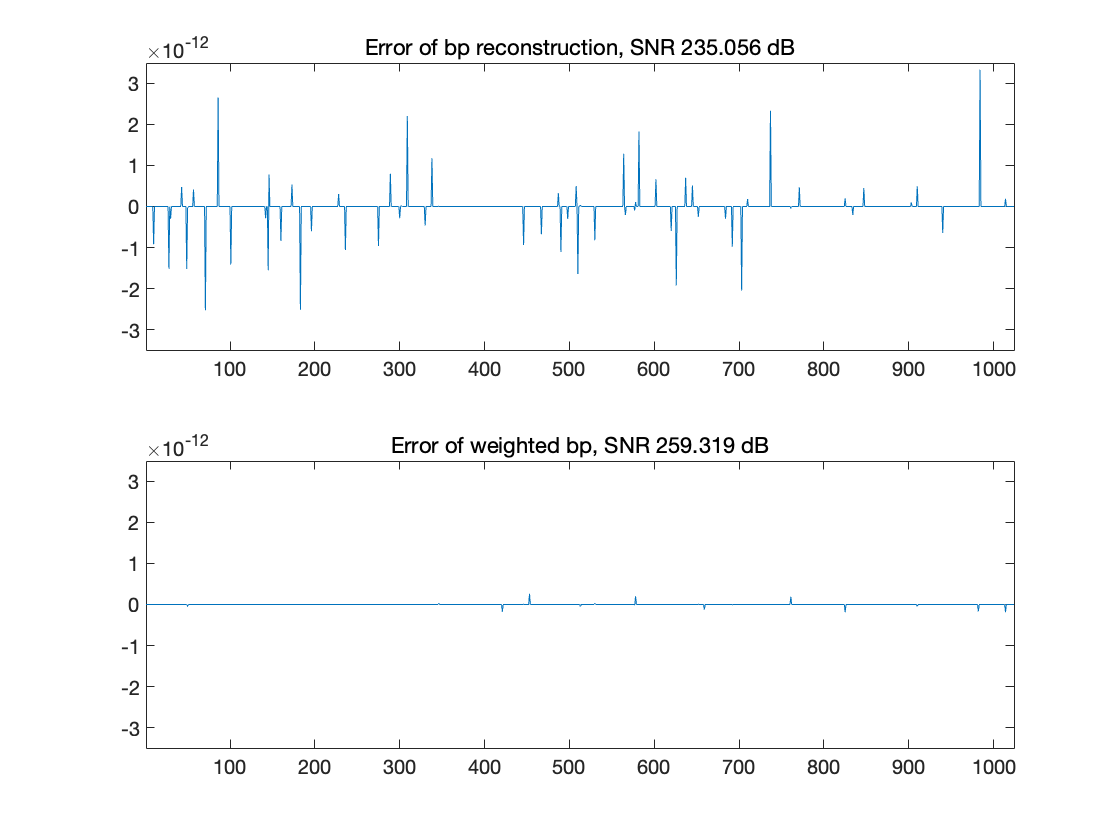}
        \caption{the error of reconstruction for Basis Pursuit and weighted Basis Pursuit with $\gamma=0.5$}
         \label{wbp2}
    \end{minipage}
\end{figure}
In Figure \ref{wbp1} and Figure \ref{wbp2}, both weighted and unweighted BP models can successfully reconstruct the signal, but the weighted BP has a better control on the error of reconstruction and a higher quantity of reconstruction according to the value of SNR.

\begin{figure}

   \centering
    \begin{minipage}{.3\textwidth}
        \centering
        \includegraphics[width=\linewidth]{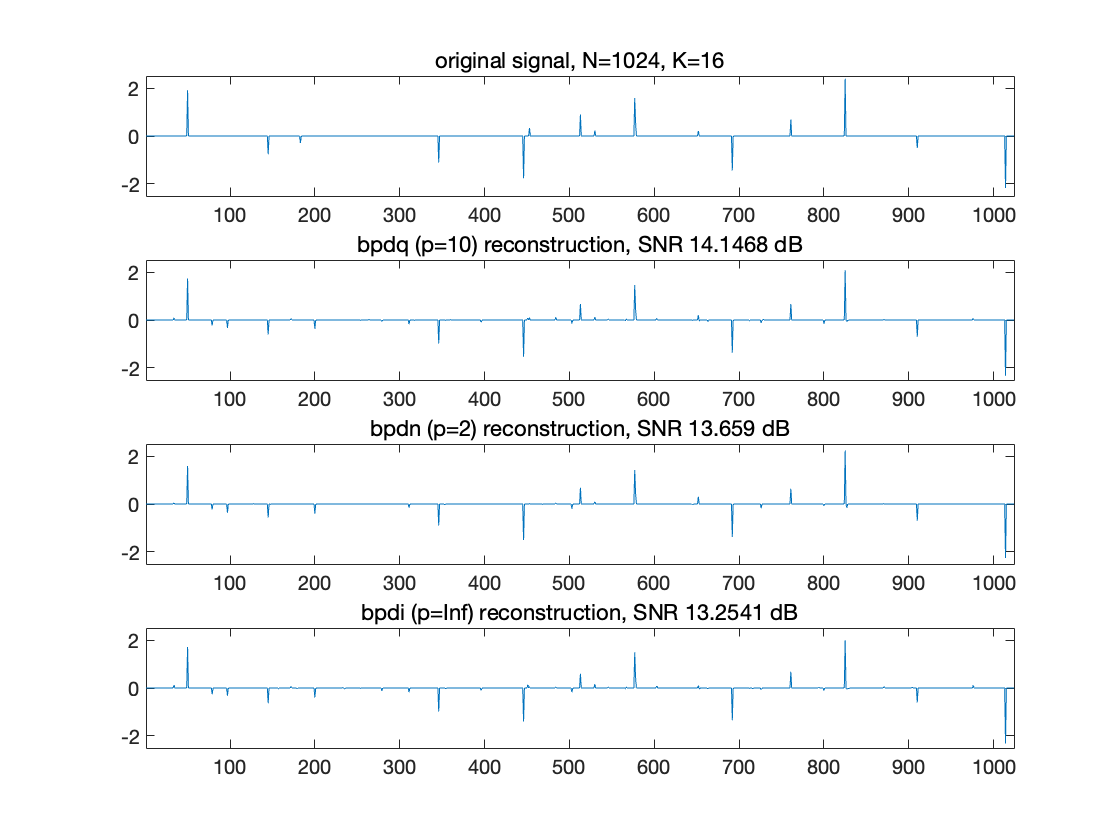}
        \caption{the original signal, the recover signals and their SNRs for $\omega$-BPDQ$_p$ with $p=10,2,\infty$}
        \label{wbpdq1}
    \end{minipage}
    \hfill 
    \begin{minipage}{.3\textwidth}
        \centering
        \includegraphics[width=\linewidth]{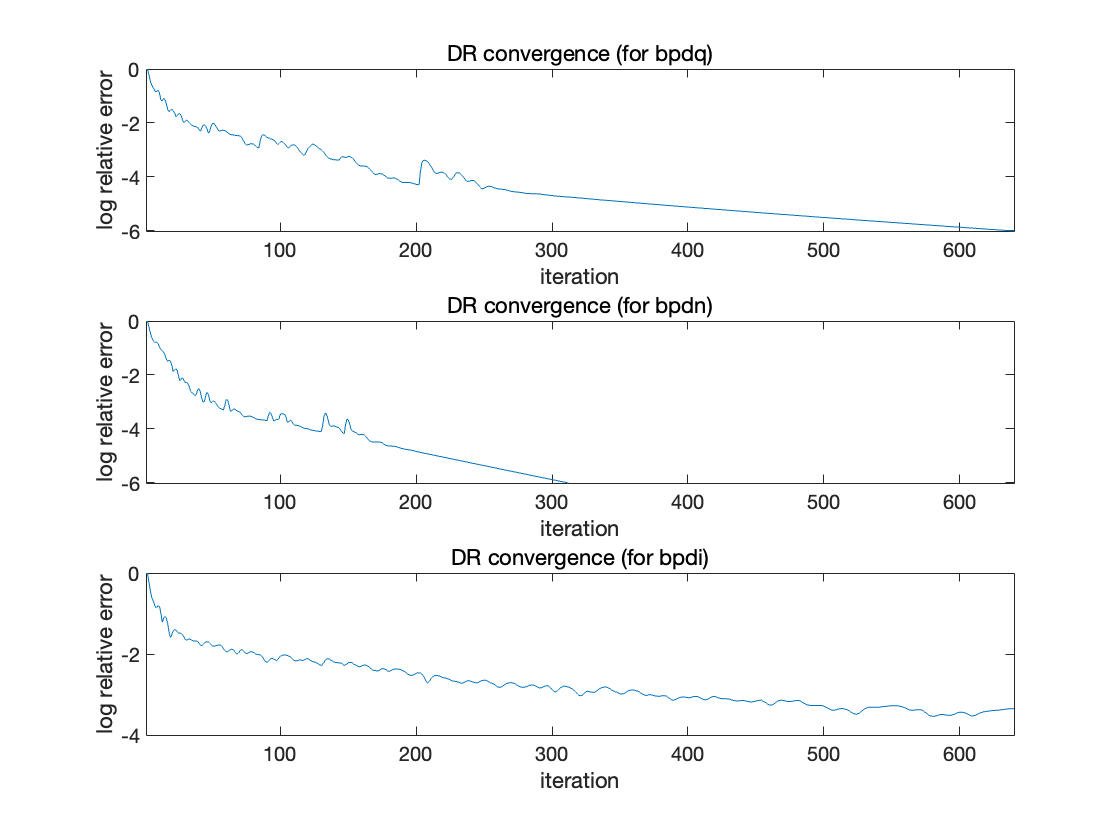}
        \caption{DR convergence for $\omega$-BPDQ$_p$ with $p=10,2,\infty$}
         \label{wbpdq2}
    \end{minipage}
    \hfill 
    \begin{minipage}{.3\textwidth}
        \centering
        \includegraphics[width=\linewidth]{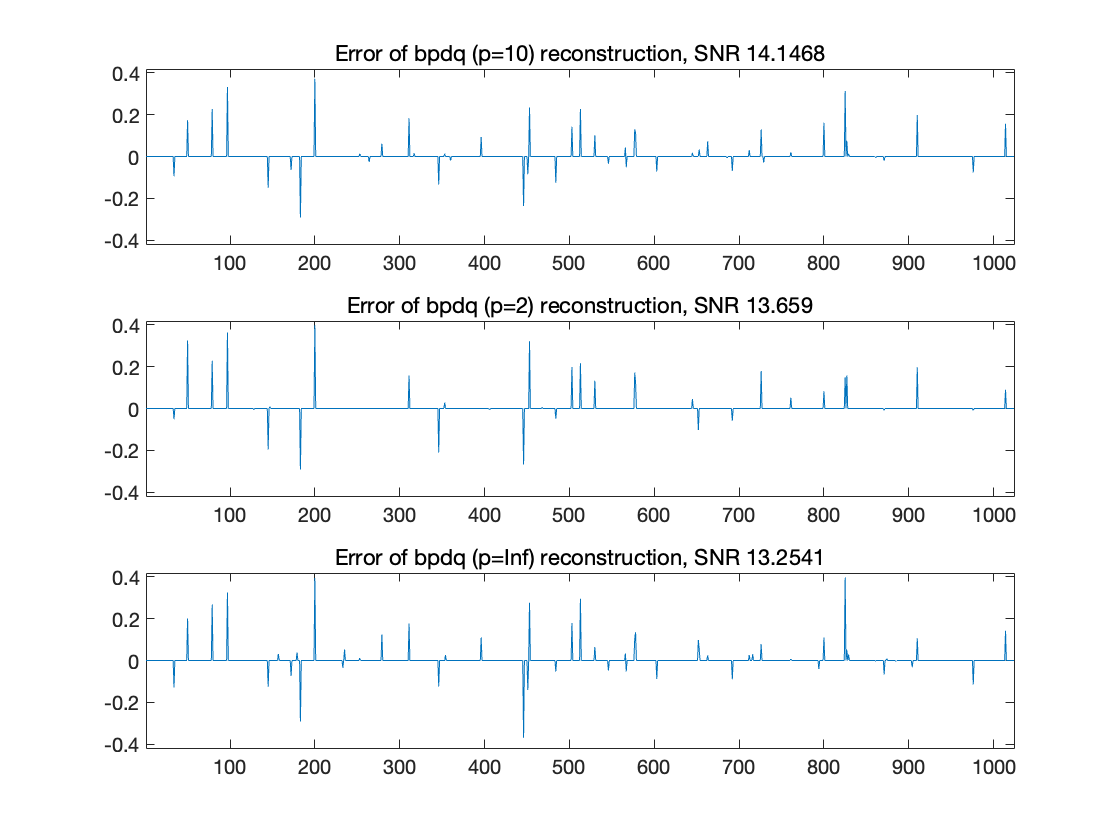}
        \caption{the error of reconstruction for $\omega$-BPDQ$_p$ with $p=10,2,\infty$}
         \label{wbpdq3}
    \end{minipage}
\end{figure}

In Figure \ref{wbpdq1} and Figure \ref{wbpdq3},  $\omega$-BPDQ$_p$ behaves better  compared to $\omega$-BPDN and  $\omega$-BPDQ$_\infty$ complied with SNR. In Figure \ref{wbpdq2}, when $p=10$ and $p=2$, DR algorithm converges to the solution after 800 iterations, but in $\omega$-BPDQ$_\infty$, the convergence is not guaranteed.

The experiments were repeated 50 times, and the mean of SNR for various values of \( p \) and \( m/K \) is plotted in Figure \ref{mean}. As observed, at higher values of \( m/K \), decoders with higher \( p \) yield better reconstruction performance. It is also evident that at lower \( m/K \), increasing \( p \) beyond a certain threshold degrades the reconstruction performance.
\begin{figure}
        \centering
        \includegraphics[width=0.5\linewidth]{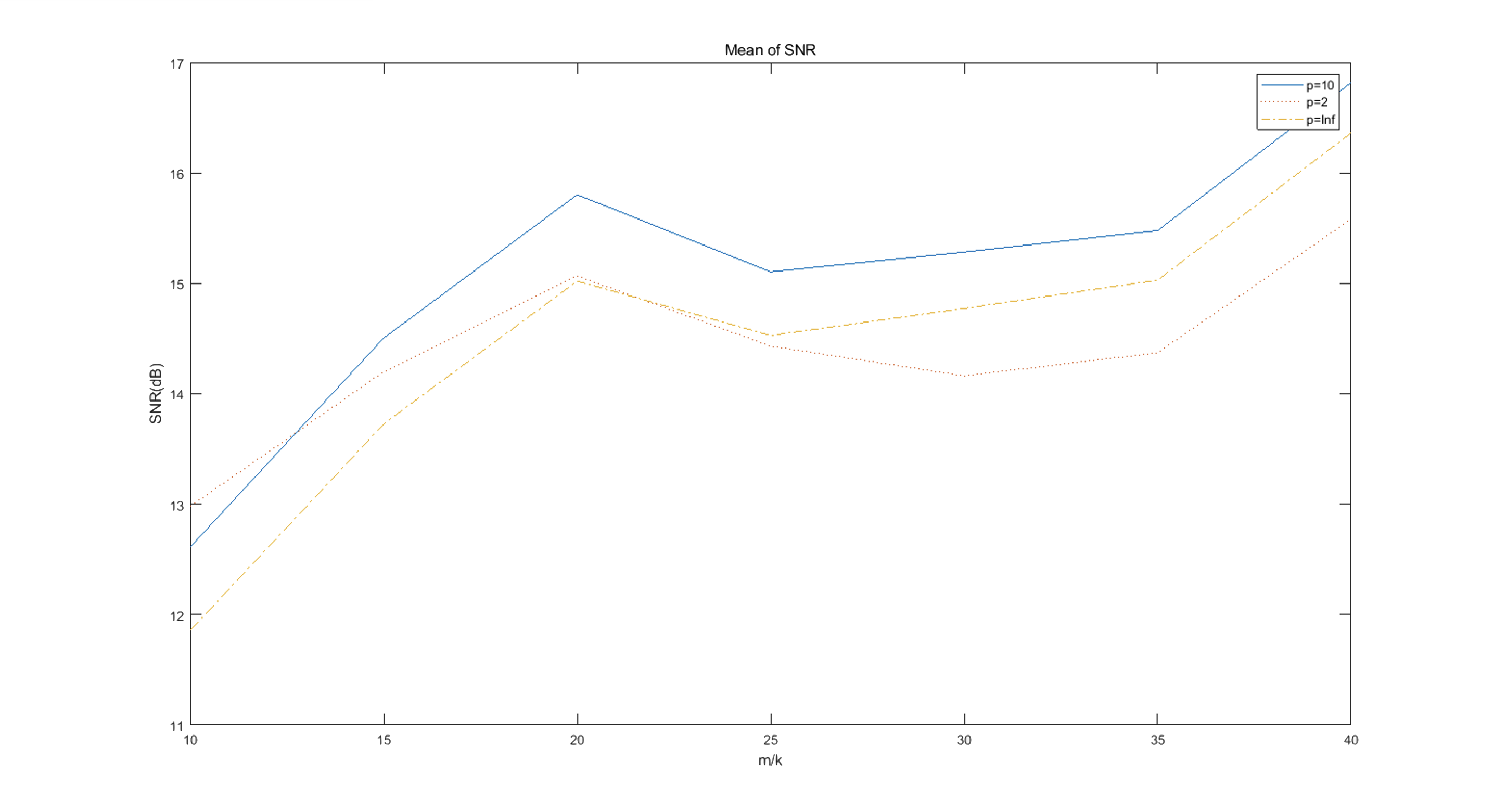}
        \caption{the mean of SNRs for $\omega$-BPDQ$_p$ with $p=10,2,\infty$}
        \label{mean}

\end{figure}

\vspace{11pt}

\newpage
\nocite{*}
\bibliographystyle{IEEEtran}
 \bibliography{cas-refs}

\vfill

\end{document}